\documentclass{IEEEtran}

\usepackage{cite}
\usepackage{amsmath,amssymb,amsfonts,mathdots}

\usepackage[]{footmisc}
\usepackage{xcolor}
\usepackage{algorithm}
\usepackage{algorithmic}
\usepackage[algo2e]{algorithm2e} 
\usepackage{graphicx}
\usepackage{textcomp}
\usepackage{psfrag}
\usepackage{epsfig}
\usepackage{subcaption}
\usepackage{float}
\usepackage{xcolor}

\usepackage{booktabs}
\usepackage{multicol}
\usepackage{arydshln}
\usepackage{multirow}
\usepackage{wrapfig}
\usepackage[numbers,sort&compress]{natbib}
\usepackage[capitalise]{cleveref}

\newtheorem{mydef}{Definition}
\newtheorem{mypro}{Proposition}
\newtheorem{myass}{Assumption}
\newtheorem{myrem}{Remark}
\newtheorem{mylem}{Lemma}
\newtheorem{mythr}{Theorem}

\begin{document}

\title{\LARGE A Unified Framework for Robust Data-Driven Model Predictive Control with Terminal Ingredients}

\author{Zhaohua~Yang,~Nan~Bai,~Pengyu~Wang,~and~Ling~Shi,~\IEEEmembership{Fellow,~IEEE}
	\thanks{\rm Z. Yang, N. Bai, P. Wang, and L. Shi are with the Department of Electronic and Computer Engineering, Hong Kong University of Science and
		Technology, Clear Water Bay, Hong Kong SAR (email: zyangcr@connect.ust.hk; eenanbai@ust.hk; pwangat@connect.ust.hk; eesling@ust.hk). L. Shi is also with the Department of Chemical and Biological Engineering and the Cheng Kar-Shun Robotics Institute (CKSRI), Hong Kong University of Science and Technology, Clear Water Bay, Hong Kong SAR.}
	\thanks{}
  }
\maketitle

\begin{abstract}
This paper proposes a unified framework for robust data-driven model predictive control (MPC) with terminal ingredients for unknown linear time-invariant systems. Based on Willems’ fundamental lemma, a single offline input-output trajectory is used as an implicit model, while online implementation only requires past input-output measurements. Bounded measurement noise is considered in both offline and online output data. The framework consists of an offline stage, where robust terminal ingredients are designed directly from noisy data, and an online stage, where these ingredients are incorporated into a one-step MPC scheme. Recursive feasibility of the proposed framework and practical exponential stability of the closed-loop system are established. Numerical results illustrate the performance tuning capability of the terminal set, with particularly pronounced benefits for short prediction horizons, and show that the proposed framework achieves performance comparable to data-driven MPC without terminal ingredients while requiring a reduced prediction horizon and less online computation.
\end{abstract}


\IEEEpeerreviewmaketitle
\section{Introduction}
Model predictive control (MPC) has become one of the most successful advanced control methodologies due to its ability to explicitly handle system constraints while achieving desirable closed-loop performance \cite{rawlings2020model, mayne2014model}. Over the past decades, a comprehensive theoretical foundation has been developed for MPC, including recursive feasibility, closed-loop stability, constraint satisfaction, and performance guarantees \cite{rawlings2020model,mayne2014model}. In particular, terminal ingredients have been extensively investigated to establish closed-loop stability for constrained systems \cite{kohler2019nonlinear}, while robust MPC techniques have been developed to maintain stability and constraint satisfaction in the presence of disturbances and model uncertainties \cite{kohler2020computationally}. These theoretical developments, together with advances in numerical optimization, have made MPC a well-established framework for the control of constrained dynamical systems.

Despite its theoretical maturity and practical success, conventional MPC fundamentally relies on an explicit model of the controlled system to predict its future behavior. In practical applications, however, constructing an accurate model can be difficult, time-consuming, or even impractical. This has motivated the development of data-driven control approaches that aim to design controllers directly from measured data, thereby bypassing system identification. For linear time-invariant (LTI) systems, an important theoretical foundation for this paradigm is Willems' fundamental lemma \cite{willems2005note}, which shows that, under a persistent excitation condition, all admissible trajectories of an unknown LTI system can be parametrized by a Hankel matrix constructed from a single measured input-output trajectory. This behavioral characterization provides an appealing alternative to model-based prediction and has become a key tool for the development of data-driven control methods \cite{markovsky2008data,de2019formulas,van2020data,berberich2020trajectory}.

Based on behavioral representations of system trajectories, an early data-driven predictive control scheme was proposed by Yang and Li \cite{yang2015data}, where a reduced Hankel matrix constructed from input-output data was employed for prediction and controller design. Coulson et al. subsequently proposed the data-enabled predictive control (DeePC) framework \cite{coulson2019data}, which directly optimizes over input-output trajectories parametrized by measured data. However, closed-loop properties such as recursive feasibility and stability were not established in these early data-driven predictive control schemes. To address this issue, Berberich et al. \cite{berberich2020data} incorporated terminal equality constraints into data-driven MPC and established recursive feasibility and exponential stability in the noise-free case, as well as practical exponential stability under bounded output measurement noise. The design of more general terminal ingredients was later investigated in \cite{berberich2021design}, where a terminal cost and a terminal set were constructed directly from noise-free input-output data to guarantee exponential stability. More recently, Bongard et al. \cite{bongard2022robust} established stability and robustness guarantees for a data-driven MPC scheme without terminal ingredients, where closed-loop stability is ensured by choosing a sufficiently long prediction horizon. Beyond these fundamental developments, data-driven MPC has also been extended to specialized settings, including changing setpoint tracking, resilient control, and event-triggered control \cite{berberich2020data_ifac,liu2022data,deng2024event}. A recent overview of systems-theoretic guarantees and developments in data-driven MPC can be found in \cite{berberich2025overview}.

Despite these advances, a gap remains between robustness analysis and terminal ingredient design. In particular, systems-theoretic guarantees for robust data-driven MPC with general terminal ingredients have not been established. Such an analysis is particularly important since MPC with and without terminal ingredients represent two classical stabilization philosophies. The former establishes stability by incorporating suitable terminal costs and terminal sets, whereas the latter relies on a sufficiently long prediction horizon, as pursued in \cite{bongard2022robust}. In this sense, a robustness analysis of data-driven MPC with general terminal ingredients provides a complementary counterpart to the terminal-free framework in \cite{bongard2022robust}.

Establishing such guarantees, however, introduces several nontrivial challenges. First, suitable terminal ingredients need to be designed directly from noisy data while remaining valid for the unknown underlying system. Although \cite{berberich2021design} provides a data-driven design of terminal ingredients, its main theoretical development assumes noise-free measurements, while the noisy-data extension discussed therein relies on an additional nontrivial condition on the measurement noise. Second, recursive feasibility requires constructing a feasible solution at the next time step from the current optimal solution. Such a construction needs to be carefully coordinated with the properties of the terminal ingredients. Finally, the resulting candidate solution must lead to a suitable Lyapunov decrease condition. While this step follows relatively standard arguments in the noise-free case, measurement noise causes discrepancies between predicted and actual trajectories and substantially complicates the closed-loop stability analysis.

Motivated by these observations, this paper develops a robust data-driven MPC framework with terminal ingredients and establishes rigorous closed-loop guarantees. The main contributions are summarized as follows.
\begin{enumerate}
    \item We develop a unified robust data-driven MPC framework with general terminal ingredients. The proposed scheme consists of two stages. Offline, a terminal controller, a terminal cost, and a terminal set are designed directly from noisy input-output data. Online, these terminal ingredients are incorporated into a regularized data-driven MPC problem, which is implemented in a standard one-step fashion (\cref{robust terminal set design theorem,full algorithm}).
    \item We propose a new candidate solution construction for the robustness analysis of data-driven MPC. Unlike existing constructions \cite{berberich2020data,bongard2022robust}, both the optimal input and output trajectories are simultaneously shifted, after which the terminal controller is used to append the terminal input-output pair and a slack variable is introduced to absorb the mismatch induced by noisy data. Based on this construction, we establish recursive feasibility under one-step implementation and prove practical exponential stability of the resulting closed-loop system (\cref{recursive feasibility proposition,practical exponential stability theorem}).
    \item Numerical studies on an inverted pendulum model illustrate three features of the proposed framework. First, the terminal set level provides a tuning mechanism for balancing convergence speed and closed-loop performance. Second, compared with terminal equality MPC, the benefit of the terminal set is particularly obvious for short prediction horizons, whereas the two schemes exhibit similar behavior as the prediction horizon increases. Third, compared with data-driven MPC without terminal ingredients \cite{bongard2022robust}, the proposed scheme achieves comparable closed-loop performance with a shorter prediction horizon, thereby trading additional offline terminal design for reduced online computational effort.
\end{enumerate}

We organize this paper as follows. \cref{Preliminaries} introduces the preliminaries and formulates the control problem. In \cref{Main Results}, we present the proposed robust data-driven MPC scheme, develop a data-driven design of robust terminal ingredients, establish recursive feasibility and practical exponential stability, and discuss its relation to existing data-driven MPC schemes. \cref{Simulations} provides numerical simulations illustrating the influence of the terminal set size and prediction horizon and comparing the proposed scheme with data-driven MPC without terminal ingredients. \cref{Conclusion} concludes the paper.

\textit{Notations}: For a vector $x$ and a positive definite matrix $P\succ0$, define $\|x\|_P \triangleq \sqrt{x^\top P x}$, while $\|x\|$ denotes the Euclidean norm by default. For a sequence $\{x_k\}_{k=0}^{N-1}$, we define its depth-$L$ Hankel matrix
$$
H_L(x)\triangleq
\begin{bmatrix}
x_0 & x_1 & \cdots & x_{N-L}\\
x_1 & x_2 & \cdots & x_{N-L+1}\\
\vdots & \vdots & \ddots & \vdots\\
x_{L-1} & x_L & \cdots & x_{N-1}
\end{bmatrix},
$$
and
$
x_{[a,b]}\triangleq
\begin{bmatrix}
x_a^\top & x_{a+1}^\top & \cdots & x_b^\top
\end{bmatrix}^{\top}.
$
For a set $\mathbb X$, $\mathbb X^l$ denotes its $l$-fold Cartesian product, and $\mathrm{int}(\mathbb X)$ denotes its interior. The set of integers in an interval $[a,b]$ is denoted by $\mathbb I_{[a,b]}$. The symbol \(\mathbf 1\) denotes the all-ones vector, with dimensions labeled when necessary. The symbols $(\cdot)^\dagger$, $\mathrm{blkdiag}(\cdot)$, and $\otimes$ denote the Moore--Penrose pseudoinverse, block-diagonal concatenation, and Kronecker product, respectively. For vectors $x_1,\ldots,x_r$, $\mathrm{col}(x_1,\ldots,x_r)$ denotes the column space spanned by these vectors, while $\ker(A)$ denotes the right null space of a matrix $A$. For a symmetric matrix $P$, $\lambda_{\min}(P)$ and $\lambda_{\max}(P)$ denote its minimum and maximum eigenvalues, respectively. For two symmetric matrices $P_1$ and $P_2$, define
$
\lambda_{\min}(P_1,P_2)
\triangleq
\min\{\lambda_{\min}(P_1),\lambda_{\min}(P_2)\},
$
and similarly for
$
\lambda_{\max}(P_1,P_2).
$
For compatible matrices $A$ and $B$, the operator $\circledast$ is defined as
$
A\circledast B\triangleq ABA^\top.
$
Finally, the symbol $\mathcal K_\infty$ denotes the class of continuous functions $\alpha:\mathbb R_{\geq 0}\rightarrow\mathbb R_{\geq 0}$ that are strictly increasing, unbounded, and satisfy $\alpha(0)=0$.

\section{Preliminaries}\label{Preliminaries}
\subsection{System Description}
In this paper, we consider a discrete-time LTI system:
\begin{equation}\label{basic system}
	\begin{aligned}
		x_{k+1} &= Ax_k + Bu_k, \\
		y_k &= Cx_k + Du_k,
	\end{aligned}
\end{equation}
where $x_k \in \mathbb{R}^{n_x}$, $u_k \in \mathbb{R}^{n_u}$, and $y_k \in \mathbb{R}^{n_y}$ denote the system state, control input, and system output at time step $k$, respectively. We make the standard assumption that $(A,B)$ is controllable and $(A,C)$ is observable. Throughout this paper, we assume that the system matrices $A$, $B$, $C$, and $D$ are unknown, but we can access the noise-free control input $u_k$ and noisy system output $\Tilde{y}_k$ satisfying $\Tilde{y}_k = y_k + \varepsilon_k$. The term $\varepsilon_k$ represents an unknown but bounded measurement noise whose bound will be specified later. Our goal is to design a data-driven controller that stabilizes a prescribed
equilibrium while satisfying input constraints. The formal problem formulation will be provided in \cref{problem formulation}. 

\subsection{Willems et al.'s Fundamental Lemma}
In model-based MPC, the system dynamics are explicitly used to characterize the set of admissible input-output trajectories. However, in a data-driven setting where the system model is unknown, such a characterization must be obtained directly from measured data.
Willems et al.'s fundamental lemma provides a key result in this direction, showing that, for LTI systems, the column space of a Hankel matrix constructed from a single persistently exciting trajectory spans the set of all possible trajectories of the system. Before presenting the fundamental lemma, we first introduce the concept of persistent excitation.
\begin{mydef}
	We say a sequence of signals $\{u_k\}_{k=0}^{N-1}$ is persistently exciting of order $L$ if $H_L(u)$ has full row rank.
\end{mydef}

With the concept of persistent excitation, we can now present the fundamental lemma.
\begin{mylem}[\cite{willems2005note}]\label{lemma: Fundamental Lemma}
Suppose that $\{u^d_k,y^d_k\}_{k=0}^{N-1}$ is a trajectory of \eqref{basic system} and that $u^d$ is persistently exciting of order $L+n_x$. Then, $\{\Bar{u}_k,\Bar{y}_k\}_{k=0}^{L-1}$ is a trajectory of \eqref{basic system} if and only if there exists a vector $g\in \mathbb{R}^{N-L+1}$ such that 
\begin{equation}
	\begin{bmatrix}
		H_L(u^d) \\
		H_L(y^d)
	\end{bmatrix}g = \begin{bmatrix}
		\Bar{u} \\
		\Bar{y}
	\end{bmatrix}.
\end{equation}
\end{mylem}

Willems et al.'s fundamental lemma has received increasing attention in recent years and has been applied in various areas, such as (1) controller synthesis \cite{de2019formulas}, (2) moving horizon estimation \cite{wolff2024robust}, and (3) MPC \cite{coulson2019data,berberich2020data}. This is mainly because it provides an appealing parametrization of input-output trajectories directly from offline data, thereby avoiding the need to first identify system matrices and subsequently apply model-based MPC methods. The fundamental lemma serves as the theoretical foundation of this paper, upon which we aim to develop a data-driven MPC scheme that achieves the desired stability properties.

\subsection{Extended System Representation}
We now introduce the extended system representation, which is a key concept in our proposed data-driven MPC scheme. We first define the following matrices:
\begin{equation}
\begin{aligned}
    \mathcal{R}_l &\triangleq \begin{bmatrix}
		A^{l-1}B & A^{l-2}B & \cdots & B
	\end{bmatrix},\\
    \mathcal{O}_l &\triangleq 
	\begin{bmatrix}
		C \\
		CA \\
		\vdots \\
		CA^{l-1}
	\end{bmatrix},\mathcal{T}_l \triangleq \begin{bmatrix}
		D & 0 & \cdots & 0 \\
		CB & D & \cdots & 0 \\
		\vdots & \vdots & \ddots & \vdots \\
		C A^{l-2}B & C A^{l-3}B & \cdots & D
	\end{bmatrix}.
\end{aligned}
\end{equation}
Since $(A,C)$ is observable, we can define the observability index $\underline{l}$ as the smallest integer $l$ such that the observability matrix $\mathcal{O}_l$
has full column rank. To avoid introducing extra notation, we abuse notation and simply take $l$ to denote an upper bound on the observability index, i.e., $l\geq \underline{l}$, which is assumed to be available. The extended state is then defined as
\begin{equation}
	\xi_k \triangleq \begin{bmatrix}
		u_{[k-l,k-1]} \\
		y_{[k-l,k-1]}
	\end{bmatrix} \in \mathbb{R}^{n_\xi}, n_\xi = l(n_u+n_y).
\end{equation}
The extended system representation is given by
\begin{equation}\label{extended system}
	\begin{aligned}
		\xi_{k+1} &= \Tilde{A} \xi_k + \Tilde{B} u_k, \\
		y_k &= \Tilde{C} \xi_k + \Tilde{D} u_k,
	\end{aligned}
\end{equation}
where $\tilde A$, $\tilde B$, $\tilde C$, and $\tilde D$ denote the extended-system matrices defined in \cite{berberich2021design}.
The input-output behavior of \eqref{extended system} is known to be equivalent to that of \eqref{basic system} \cite{goodwin2014adaptive}. By observability of $(A,C)$ and the system dynamics, we obtain the following relation between the extended state $\xi_k$ and the system state $x_k$, which will be useful in the subsequent theoretical analysis
\begin{equation}\label{relationship between two systems 2}
	x_k = \Gamma\xi_k,\Gamma \triangleq \begin{bmatrix}\mathcal{R}_l-A^l\mathcal{O}_l^\dagger\mathcal{T}_l &A^l\mathcal{O}_l^\dagger\end{bmatrix}.
\end{equation}
The following lemma summarizes an interesting property of the extended system representation.
\begin{mylem}[\cite{bongard2022robust}]\label{IOSS lemma}
	The extended system \eqref{extended system} is input-output-to-state stable (IOSS), i.e., there exists an IOSS Lyapunov function $W(\xi) = \xi^\top P_W \xi$ with $P_W\succ 0$ such that
	\begin{equation}\label{IOSS eq}
		W(\xi_{k+1}) - W(\xi_k) \leq -\frac{1}{2}\|\xi_k\|^2 + c_{W,1}\|u_k\|^2 + c_{W,2}\|y_k\|^2,
	\end{equation}
	for some positive constants $c_{W,1}>0$ and $c_{W,2}>0$.
\end{mylem}

\subsection{Problem Formulation}\label{problem formulation}
Before formally formulating the problem, we first introduce two important concepts.
\begin{mydef}
	We say that the pair $(u^s,y^s)$ is an equilibrium of \eqref{basic system} if $\{\Bar{u}_k,\Bar{y}_k\}_{k=0}^{\underline{l}}$ is a trajectory of \eqref{basic system} with $(\Bar{u}_k, \Bar{y}_k) = (u^s, y^s)$ for all $k\in \mathbb{I}_{[0,\underline{l}]}$.
\end{mydef}

The use of a trajectory of length $\underline{l}+1$ is motivated by observability. In particular, if the system has observability index $\underline{l}$, then the input-output data $\{u_k,y_k\}_{k=0}^{\underline{l}}$ satisfy a linear relation, where $y_{\underline{l}}$ is uniquely determined by $\{u_k,y_k\}_{k=0}^{\underline{l}-1}$ and $u_{\underline{l}}$. Requiring that the sequence $\{u_k=u^s,\ y_k=y^s\}_{k=0}^{\underline{l}}$ is a valid trajectory ensures that this linear relation is satisfied at the steady pair. This effectively fixes the underlying linear relation at a constant point, ensuring that repeated application of $u^s$
 yields the same output $y^s$. 
\begin{mydef}
	For system \eqref{extended system}, an equilibrium $(u^s,y^s)$ corresponding to a constant input $u^s$ is said to be practically exponentially stable if there exist $c>0$, $\rho \in (0,1)$, and $d\in \mathcal{K}_\infty$ such that for any initial state $\xi_0$, the extended system state satisfies
	\begin{equation}\label{practical exponential stability}
	\|\xi_k - \xi^s\|_2 \le c \rho^k \|\xi_0 - \xi^s\|_2 + d(\bar{\varepsilon}), \quad k \in \mathbb{I}_{\ge 0},
	\end{equation}
	where
$\xi^s \triangleq \left[\begin{smallmatrix}
    &\mathbf{1}_l\otimes u^s\\
    &\mathbf{1}_l\otimes y^s
\end{smallmatrix}\right]$ is the extended equilibrium state and $\bar{\varepsilon}$ is the noise bound defined later. If $\bar{\varepsilon}=0$, then $(u^s,y^s)$ is said to be exponentially stable.
\end{mydef}

We can now formally formulate the problem considered in this paper as follows: \textit{Given a desired equilibrium $(u^s,y^s)$, design a robust data-driven MPC scheme that ensures practical exponential stability in the sense of \eqref{practical exponential stability} while satisfying the input constraints.}

\section{Main Results}\label{Main Results}
In this section, we present a robust data-driven MPC scheme with terminal ingredients. Our main contribution is the theoretical result that ensures the practical exponential stability of the closed-loop system under the proposed MPC scheme. In \cref{Robust Data-Driven MPC with Terminal ingredients}, we present the proposed robust data-driven MPC scheme with terminal ingredients. In \cref{Robust Terminal Controller Design}, we design robust terminal ingredients directly from offline data, ensuring that specific terminal conditions are satisfied. In \cref{Lower and Upper bound of Lyapunov Function}, we establish the lower and upper bounds of the Lyapunov function. In \cref{Bound on the Output Prediction Error}, we explicitly bound the difference between the predicted output and the nominal output with the application of the predicted input. In \cref{Analysis of the Optimal Solution Continuity}, we provide some technical results on the continuity of the optimal solution. In \cref{Recursive Feasibility}, we construct a feasible solution for the next time step by simultaneously shifting the current optimal input and output trajectories and thus prove the recursive feasibility. In \cref{Practical Exponential Stability}, we present the main theoretical result of practical exponential stability. Finally, in \cref{remark on different mpc scheme}, we discuss the relation between the proposed framework and existing data-driven MPC schemes.

\subsection{Robust Data-Driven MPC with Terminal Ingredients}\label{Robust Data-Driven MPC with Terminal ingredients}
Assume that a set of offline input-output data $\{u_k^d,\tilde{y}_k^d\}_{k=0}^{N-1}$ has been collected, where $\tilde{y}_k^d = y_k^d + \varepsilon_k^d$ denotes the output data corrupted by offline measurement noise. The superscript $d$ is used to indicate offline data. In contrast, MPC requires real-time measurement of the system output. Specifically, at each time step $k$, the measured output is given by $\tilde{y}_k = y_k + \varepsilon_k$, which is subject to online measurement noise. The measurement noise is assumed to be unknown but bounded, and we consider a general case where the noise bounds for the offline and online data may differ, i.e., $\|\varepsilon^d_k\|_2 \leq \bar{\varepsilon}_{off}$ and $\|\varepsilon_k\|_2 \leq \bar{\varepsilon}_{on}$, where $\bar{\varepsilon}_{off}$ and $\bar{\varepsilon}_{on}$ are known constants. To facilitate the theoretical analysis, we define $\bar{\varepsilon} = \max\{\bar{\varepsilon}_{off},\bar{\varepsilon}_{on}\}$ as the overall noise bound. The proposed robust data-driven MPC problem, which must be solved at each time step $t$, is formulated as follows:
\begin{subequations}\label{robust data-driven mpc scheme}
	\begin{align}
		&\Bar{J}^*_L(u_{[t-l,t-1]},\tilde{y}_{[t-l,t-1]})\nonumber \\ 
		&\begin{aligned}&= \min_{\substack{\Bar{\alpha}(t),\Bar{\sigma}(t),\\ \Bar{u}(t),\Bar{y}(t)}} \sum_{k=0}^{L-1}\left(\|\Bar{u}_k(t)-u^s\|_R^2 + \|\Bar{y}_k(t)-y^s\|_Q^2\right) \\
		&+ \|\Bar{\xi}_L(t)-\xi^s\|_P^2  + \lambda_\alpha \Bar{\varepsilon}^{\beta_\alpha}\|\Bar{\alpha}(t)\|^2+ \frac{\lambda_\sigma}{\Bar{\varepsilon}^{\beta_\sigma}}\|\Bar{\sigma}(t)\|^2\end{aligned} \label{objective function} \\
        &\qquad \mathrm{s.t.}\quad \begin{bmatrix}
            \Bar{u}(t)\\
            \Bar{y}(t) + \Bar{\sigma}(t)
        \end{bmatrix} = \begin{bmatrix}
            H_{L+l}(u^d)\\
            H_{L+l}(\Tilde{y}^d)
        \end{bmatrix}\Bar{\alpha}(t), \label{parametrization using offline data}\\
        &\qquad\qquad\ \begin{bmatrix}
            \Bar{u}_{[-l,-1]}(t)\\
            \Bar{y}_{[-l,-1]}(t)
        \end{bmatrix} = \begin{bmatrix}
            u_{[t-l,t-1]}\\
            \Tilde{y}_{[t-l,t-1]}
        \end{bmatrix}, \label{initialization}\\
        &\qquad\qquad\ \Bar{\xi}_L(t)\in \Sigma_f, \Bar{\xi}_L(t) = \begin{bmatrix}
            \Bar{u}_{[L-l,L-1]}(t)\\
            \Bar{y}_{[L-l,L-1]}(t)
        \end{bmatrix}, \label{terminal ingredient}\\
        &\qquad\qquad\ \Bar{u}_k(t)\in\mathbb{U}, k\in\mathbb{I}_{[0,L-1]}. \label{input constraints}
	\end{align}
\end{subequations}
There are some important remarks regarding the above robust data-driven MPC problem. 
\begin{enumerate}
	\item This data-driven MPC problem requires the past $l$ steps of input-output data for initialization. In the rest of the paper, we will use $\Bar{J}_L^*(\Tilde{\xi}_t)$ to denote the optimal cost for simplicity, where we define $\Tilde{\xi}_t \triangleq \left[\begin{smallmatrix}
	    u_{[t-l,t-1]}^\top &\Tilde{y}_{[t-l,t-1]}^\top
	\end{smallmatrix}\right]^\top$.
	\item Due to the presence of measurement noise, $(\bar{u}(t), \bar{y}(t))$ is generally not a system trajectory. To tackle this issue, we introduce a slack variable $\Bar{\sigma}(t)$ in \eqref{parametrization using offline data} to absorb the mismatch and penalize it in the objective function \eqref{objective function}. We also penalize the decision variable $\Bar{\alpha}(t)$ in the objective to reduce the influence of measurement noise.
	\item We require $Q,R\succ0$ and $\lambda_\alpha,\lambda_\sigma,\beta_\alpha,\beta_\sigma>0$. Note that a terminal cost term with weight matrix $P \succ 0$ is incorporated, which plays a crucial role in stability analysis in the presence of a terminal set constraint. 
	\item The input constraint set $\mathbb{U}$ in \eqref{input constraints} can be any convex polytope (ellipsoid) that satisfies $u^s\in \mathrm{int}(\mathbb{U})$.
    \item If $\bar{\varepsilon}=0$, Problem \eqref{robust data-driven mpc scheme} is understood in the limiting sense by imposing \(\bar\sigma=0\) and removing the two regularization terms.
\end{enumerate}

To align with the structure of Problem \eqref{robust data-driven mpc scheme}, we make the following assumption:
\begin{myass}\label{persistently exciting assumption}
	The input sequence $u^d$ is persistently exciting of order $L+l+n_x$.
\end{myass}

Note that in this assumption, the order of persistent excitation is different from that in \cref{lemma: Fundamental Lemma}, where only $L+n_x$ is required. This is because, in the proposed MPC scheme, the past $l$ steps of input-output data are used for initialization. The length of the predicted trajectory is thereby $L+l$.

We next make an assumption regarding the prediction horizon $L$:
\begin{myass}\label{prediction horizon assumption}
	The prediction horizon $L$ satisfies $L\geq l+n_x$.
\end{myass}

This assumption sets a lower bound on the prediction horizon $L$, which allows us to leverage the controllability of \eqref{basic system} for trajectory construction of length $L-l$. 

We assume that the terminal set $\Sigma_f$ satisfies the following conditions:
\begin{myass}\label{terminal conditions assumption}
	There exist matrices $P\succ 0$ and $K\in \mathbb{R}^{n_u\times n_\xi}$ and a set $\Sigma_f\subseteq 
\mathbb{U}^l\times (\mathbb{R}^{n_y})^l$ satisfying $\xi^s\in \Sigma_f$ such that, for every $\xi\in \Sigma_f$, the following properties hold with $u = u^s + K(\xi-\xi^s)$, $\xi^e = \xi-\xi^s$, $u^e = u-u^s$, $y^e = y-y^s$:
\begin{enumerate}
	\item $\Tilde{A}\xi + \Tilde{B}u \in \Sigma_f$;
	\item $u\in \mathbb{U}$;
	\item $\|(\Tilde{A}+\Tilde{B}K)\xi^e\|^2_P \leq \|\xi^e\|^2_P - \|u^e\|^2_R - \|y^e\|^2_Q$.
\end{enumerate}
\end{myass}
\cref{terminal conditions assumption} is a standard assumption in model-based MPC to ensure the stability of the closed-loop system \cite{chen1998quasi,rawlings2020model}. In \cref{Robust Terminal Controller Design}, we will show how to design the weight matrix $P$, terminal controller $K$, and terminal set $\Sigma_f$ directly from the offline input-output data $\{u_k^d,\tilde{y}_k^d\}_{k=0}^{N-1}$, which is a key ingredient of this paper. It is straightforward to verify that the proposed MPC scheme \eqref{robust data-driven mpc scheme} is a convex quadratically constrained quadratic program (QCQP) if $\mathbb{U}$ and $\Sigma_f$ are convex ellipsoids. Therefore, it can be solved efficiently using off-the-shelf tools, e.g., CVX \cite{grant2008cvx}. We denote the optimal solution to Problem \eqref{robust data-driven mpc scheme} at time step $t$ by $\Bar{\alpha}^*(t)$, $\Bar{\sigma}^*(t)$, $\Bar{u}^*(t)$, and $\Bar{y}^*(t)$. The complete control algorithm based on Problem \eqref{robust data-driven mpc scheme} is summarized in \cref{full algorithm}.

\begin{algorithm}[t]
    \caption{Robust data-driven MPC algorithm}
    \textbf{Input:} the offline input-output data $\{u_k^d,\tilde{y}_k^d\}_{k=0}^{N-1}$, where $u^d$ satisfies \cref{persistently exciting assumption}; the prediction horizon $L$ satisfying \cref{prediction horizon assumption}; the equilibrium $(u^s,y^s)$; the input constraint set $\mathbb{U}$; the parameters $Q,R,\lambda_\alpha,\lambda_\sigma,\beta_\alpha,\beta_\sigma$.\\
    \textbf{Part 1 (Robust terminal set design):} Design the terminal set $\Sigma_f$, terminal controller $K$, and weight matrix $P$ from the offline data.\\
    \textbf{Part 2 (One-step data-driven MPC):}\\
    (1) At each time step $t$, measure the past $l$ steps of input-output data $(u_{[t-l,t-1]},
		\Tilde{y}_{[t-l,t-1]})$ and solve Problem \eqref{robust data-driven mpc scheme}.\\
    (2) Apply $u_t = \Bar{u}^*_0(t)$.\\
    (3) Set $t = t + 1$ and go back to (1).

    \label{full algorithm}
\end{algorithm}

\subsection{Robust Terminal Set Design}\label{Robust Terminal Controller Design}
In this section, we focus on the first part of \cref{full algorithm}, i.e., the design of the terminal set $\Sigma_f$, terminal controller $K$, and weight matrix $P$ from the offline input-output data $\{u_k^d,\tilde{y}_k^d\}_{k=0}^{N-1}$. Before presenting the design procedure, we first introduce a technical lemma that will be useful in the subsequent analysis.
\begin{mylem}[{\cite{bisoffi2021trade}}]\label{s procedure}
Let $T_0,\ldots,T_\ell\in \mathbb{R}^{(q+p)\times(q+p)}$ be symmetric matrices. If there exist scalars $\tau_1\ge0,\ldots,\tau_\ell\ge0$ satisfying
$
T_0-\sum_{i=1}^{\ell}\tau_iT_i\succeq0,
$
then
$
\left[\begin{smallmatrix} I \\ Z \end{smallmatrix}\right]^\top T_0 \left[\begin{smallmatrix} I \\ Z \end{smallmatrix}\right] \succeq 0
$
holds for every $Z\in\mathbb{R}^{p\times q}$ satisfying
$
\left[\begin{smallmatrix} I \\ Z \end{smallmatrix}\right]^\top T_i \left[\begin{smallmatrix} I \\ Z \end{smallmatrix}\right] \succeq 0, i=1,\ldots,\ell.
$
\end{mylem}

We first partition the input-output data $\{u_k^d,\tilde{y}_k^d\}_{k=0}^{N-1}$ into the sequence $\{\Tilde{\zeta}^d_k\}^{N-1}_{k=l}$, where we define 
\begin{equation}
    \Tilde{\zeta}^d_k \triangleq \begin{bmatrix}
        u^d_{[k-l,k]}\\
        \Tilde{y}^d_{[k-l,k]}
    \end{bmatrix} \in \mathbb{R}^{n_\zeta}, n_\zeta = (l+1)(n_u+n_y).
\end{equation}
By substituting $\Tilde{\zeta}^d_k$ into \eqref{extended system}, we have the following relation:
\begin{equation}
    \Tilde{y}^d_k = 
    \begin{bmatrix}
        \begin{array}{c|c|c}
            \Delta_1 & \Delta_2 & \Delta_3
        \end{array}
    \end{bmatrix}
    \Tilde{Z}_k
    +
    \begin{bmatrix}
        \begin{array}{c|c}
            I &-\Delta_2
        \end{array}
    \end{bmatrix}
    E^d_k,
\end{equation}
where we denote
\begin{equation}
    \begin{aligned}
    \Delta_1 &\triangleq C\mathcal{R}_l-CA^l\mathcal{O}^\dagger_l\mathcal{T}_l, \Delta_2 \triangleq CA^l\mathcal{O}^\dagger_l, \Delta_3 \triangleq D,\\
    \Tilde{Z}_k &\triangleq
\begin{bmatrix}
    \begin{array}{c|c|c}
        \begin{smallmatrix}
            (u^d_{k-l})^\top & \cdots & (u^d_{k-1})^\top
        \end{smallmatrix}
        &
        \begin{smallmatrix}
            (\Tilde{y}^d_{k-l})^\top & \cdots & (\Tilde{y}^d_{k-1})^\top
        \end{smallmatrix}
        &
        \scriptstyle{(u^d_k)^\top}
    \end{array}
\end{bmatrix}^{\top},\\
E^d_k & \triangleq
\begin{bmatrix}
    \begin{array}{c|c}
        \scriptstyle{(\varepsilon^d_k)^\top}
        &
        \begin{smallmatrix}
            (\varepsilon^d_{k-l})^\top & \cdots & (\varepsilon^d_{k-1})^\top
        \end{smallmatrix}
    \end{array}
\end{bmatrix}^{\top}.
    \end{aligned}
\end{equation}
Since the system matrices $A$, $B$, $C$, and $D$ are unknown, the matrices $\Delta_1$, $\Delta_2$, and $\Delta_3$ are also unknown. However, we can derive a rough uncertainty set for $\Delta \triangleq \begin{bmatrix}
    \Delta_1 & \Delta_2 & \Delta_3
\end{bmatrix}$ based on the offline data. Recall that the offline measurement noise satisfies $\|\varepsilon^d_k\|_2 \leq \bar{\varepsilon}_{off}$, which implies that $E^d_k {E^d_k}^\top \preceq \Theta$, where $\Theta = \mathrm{blkdiag}(\Theta_{11},\Theta_{22}), \Theta_{11} = (l+1)\bar{\varepsilon}_{off}^2 I_{n_y}, \Theta_{22} = (l+1) \bar{\varepsilon}_{off}^2 I_{ln_y}$. Therefore, we have
\begin{equation}
    (\Tilde{y}^d_k - \Delta \Tilde{Z}_k)(\Tilde{y}^d_k - \Delta \Tilde{Z}_k)^\top \preceq 
    \begin{bmatrix}
        I &-\Delta_2
    \end{bmatrix}
    \Theta
    \begin{bmatrix}
        I &-\Delta_2
    \end{bmatrix}^\top,
\end{equation}
and thus 
\begin{equation}
    \begin{bmatrix}
        I\\
        \Delta^\top
    \end{bmatrix}^\top
    \underbrace{
    \begin{bmatrix}
        \Tilde{y}^d_k (\Tilde{y}_k^d)^\top - \Theta_{11} & -\Tilde{y}^d_k {\Tilde{Z}_k}^\top \\
        -\Tilde{Z}_k (\Tilde{y}_k^d)^\top & \Tilde{Z}_k {\Tilde{Z}_k}^\top - \left[
            \begin{smallmatrix}
                0 &0 &0\\
                0 &\Theta_{22} &0\\
                0 &0 &0
            \end{smallmatrix}
        \right]
    \end{bmatrix}}_{\Omega_k}
    \begin{bmatrix}
        I\\
        \Delta^\top
    \end{bmatrix}\preceq 0.
\end{equation}
We slightly modify the above inequality by left- and right-multiplying $T$ and its transpose, where we define $T \triangleq \begin{bmatrix}
    \begin{smallmatrix}
    0 &\cdots &0 &I
    \end{smallmatrix}
\end{bmatrix}^\top\in\mathbb{R}^{(n_u+n_y)l\times n_y}$. This leads to the following inequality
\begin{equation}\label{single data QMI}
    \begin{bmatrix}
        I\\
        (T\Delta)^\top
    \end{bmatrix}^\top
    \underbrace{
        \begin{bmatrix}
            T &0\\
            0 &I
        \end{bmatrix}\Omega_k
        \begin{bmatrix}
            T &0\\
            0 &I
        \end{bmatrix}^\top
    }_{\Bar{\Omega}_k}\begin{bmatrix}
        I\\
        (T\Delta)^\top
    \end{bmatrix}\preceq 0.
\end{equation}
The inequalities \eqref{single data QMI} define the following outer approximation of the models that can possibly generate the noisy offline data:
\begin{equation}
    \Sigma_\Delta = \bigcap_{k=l}^{N-1} \left\{ \Delta \,\middle|\,\begin{bmatrix}
        I\\
        (T\Delta)^\top
    \end{bmatrix}^\top
    \Bar{\Omega}_k\begin{bmatrix}
        I\\
        (T\Delta)^\top
    \end{bmatrix}\preceq 0.\right\}.
\end{equation}
We can ensure that the true $\Delta$ belongs to $\Sigma_\Delta$, but we cannot distinguish its elements. Therefore, we need to design the terminal ingredients such that \cref{terminal conditions assumption} holds for all $\Delta\in \Sigma_\Delta$.  Note that by leveraging $T$, we can express $\Tilde{A}$ and $\Tilde{B}$ as $\Tilde{A} = A'+T\Delta\left[\begin{smallmatrix}I\\0\end{smallmatrix}\right], \Tilde{B} = B'+T\Delta\left[\begin{smallmatrix}0\\I\end{smallmatrix}\right]$ for some known matrices $A'$ and $B'$. 

\begin{mythr}\label{robust terminal set design theorem}
    Suppose that there exist matrices $X\in \mathbb{R}^{n_\xi\times n_\xi}$, $Y\in \mathbb{R}^{n_u\times n_\xi}$, and scalars $\{\lambda_i\}_{i=l,l+1,\ldots,N-1}$ such that the following semidefinite program (SDP) is feasible:
    \begin{subequations}\label{robust terminal set design sdp}
        \begin{align}
        &\begin{bmatrix}
            \begin{bmatrix}
                -X &0\\
                0 &0
            \end{bmatrix}-\sum\limits_{i=l}^{N-1}\lambda_i\Bar{\Omega}_i & * &*\\
            \begin{bmatrix}(A'X+B'Y)^\top &X &Y^\top\end{bmatrix} &-X &*\\
            0 & \begin{bmatrix}Q^{\frac{1}{2}}T^\top X \\R^{\frac{1}{2}}Y\end{bmatrix} & -I
        \end{bmatrix}\preceq0, \label{robust terminal set design sdp eq1}\\
        &\begin{bmatrix}
            -X &*\\
            \begin{bmatrix}Q^{\frac{1}{2}}T^\top X \\R^{\frac{1}{2}}Y\end{bmatrix} &-I
        \end{bmatrix}\prec 0, X\succ 0, \lambda_i\ge 0, i\in \mathbb{I}_{[l,N-1]}, \label{robust terminal set design sdp eq2}
            \end{align}
    \end{subequations}
then we can compute $P = X^{-1}-T QT^\top, K = YX^{-1}$ and choose $r\geq0$ such that 
    \begin{equation}\label{expression of terminal set}
        \Sigma_f \triangleq \left\{ \xi\in\mathbb{R}^{n_\xi}\middle|\, \|\xi-\xi^s\|^2_P\leq r \right\}
    \end{equation} 
is contained in $\mathbb{U}^l\times (\mathbb{R}^{n_y})^l$. Then \cref{terminal conditions assumption} holds with the obtained $P,K$, and $\Sigma_f$.
\end{mythr}
\textit{Proof}: The proof follows a similar line to \cite[Proposition 10]{berberich2021design}, with the main difference being that the noisy-data uncertainty set $\Sigma_\Delta$ is handled using \cref{s procedure}. Let $K=YX^{-1}$ and define $Z \triangleq T Q T^\top+K^\top R K -X^{-1}$. Applying a congruence transformation and the Schur complement to \eqref{robust terminal set design sdp} gives
\begin{equation}
    \begin{bmatrix}
                -X &0\\
                0 &0
    \end{bmatrix} - \begin{bmatrix}A'+B'K \\I \\K\end{bmatrix}\circledast Z^{-1}  - \sum\limits_{i=l}^{N-1}\lambda_i\Bar{\Omega}_i\preceq 0.
\end{equation}
Invoking \cref{s procedure} and applying the Schur complement gives
\begin{equation}
(\Tilde{A}+\Tilde{B}K)^\top\circledast X^{-1} -X^{-1} + T Q T^\top + K^\top R K \preceq 0.
\end{equation}
The rest of the proof follows the same procedure as in \cite[Proposition 10]{berberich2021design}.
\hfill$\square$

\cref{robust terminal set design theorem} provides a sufficient condition for designing the terminal ingredients directly from the offline data. The feasibility of the SDP \eqref{robust terminal set design sdp} can also be checked efficiently using CVX \cite{grant2008cvx}. If the SDP is feasible, we can compute $P$, $K$, and $\Sigma_f$ accordingly. The terminal controller $K$ will be used in the subsequent theoretical analysis, and the weight matrix $P$ and terminal set $\Sigma_f$ are required in the MPC problem \eqref{robust data-driven mpc scheme}. Now, we have obtained all the ingredients required for the MPC problem \eqref{robust data-driven mpc scheme} and we can implement the second part of \cref{full algorithm}. Our next step is to analyze the stability of the closed-loop system under the proposed MPC scheme, which will be presented in the following subsections. In this paper, w.l.o.g., we assume that $u^s = 0, y^s = 0$, and $\xi^s = 0$ for simplicity. The results can be easily extended to the case of nonzero equilibrium.
\subsection{Lower and Upper Bound of Lyapunov Function}\label{Lower and Upper bound of Lyapunov Function}
In this section, we establish a technical result regarding the lower and upper bounds of the Lyapunov function. The Lyapunov function is defined as $V_t = \Bar{J}_L^*(\Tilde{\xi}_t)+\gamma W(\xi_t)$, where $\Bar{J}_L^*(\Tilde{\xi}_t)$ is the optimal cost of Problem \eqref{robust data-driven mpc scheme} and $W(\xi_t) = \xi_t^\top P_W \xi_t$ with $P_W\succ 0$ is the IOSS Lyapunov function of the extended system \eqref{extended system} introduced in \cref{IOSS lemma}. The parameter $\gamma>0$ is a constant that will be determined later. The following lemma shows that $V_t$ admits quadratic lower and upper bounds in terms of the extended state $\xi_t$ when $\xi_t$ is around the origin:
\begin{mylem}\label{lower and upper bound lemma}
    Suppose that Assumptions \ref{persistently exciting assumption}--\ref{terminal conditions assumption} hold. Then, there exists a constant $\delta>0$ such that if $\|\xi_t\|\leq \delta$, the following relation holds
    \begin{equation}\label{lower and upper bound eq}
        \gamma\lambda_{\min}(P_W)\|\xi_t\|^2 \leq V_t \leq c_{V,1}(\bar{\varepsilon})\|\xi_t\|^2 + c_{V,2}(\bar{\varepsilon}),
    \end{equation}
    where $c_{V,1}(\bar{\varepsilon})$ is lower-bounded by a positive constant as $\bar{\varepsilon}\rightarrow 0$, and $c_{V,2}(\bar{\varepsilon}) = 2l\lambda_\sigma \Bar{\varepsilon}^{2-\beta_\sigma}$.
\end{mylem} 
\textit{Proof}: See Appendix~\ref{proof of lower and upper bound lemma}.

This lemma establishes a crucial relation between the Lyapunov function $V_t$ and the extended state $\xi_t$, which will be instrumental in proving the practical exponential stability of the closed-loop system under the proposed MPC scheme. 
\subsection{Bound on the Output Prediction Error}\label{Bound on the Output Prediction Error}
Compared with the noise-free case \cite{berberich2021design}, where the optimal solution is a trajectory of the nominal system, one challenge that complicates the stability analysis in our setting is that the optimal solution $(\Bar{u}^*(t),\Bar{y}^*(t))$ of Problem \eqref{robust data-driven mpc scheme} is not necessarily a trajectory of \eqref{basic system} due to the presence of the online and offline measurement noise. We denote the nominal output trajectory starting from $x_t$ with the application of $\Bar{u}_{[0,L-1]}^*(t)$ by $\hat{y}_{[t,t+L-1]}$. Besides, we assume that there exists a sequence of constants $\rho_{2,l+k}$ satisfying 
\begin{equation}
    \rho_{2,l+k} \geq \|CA^{l+k}\mathcal{O}_l^\dagger\|^2, k\in \mathbb{I}_{[0,L-1]}.
\end{equation}
Then, we provide the following result bounding the difference between the predicted output $\Bar{y}^*$ and the nominal output $\hat{y}$. 
\begin{mylem}\label{prediction error lemma}
    Suppose that the MPC problem \eqref{robust data-driven mpc scheme} is feasible at time step $t$. Then, the following holds for $k\in\mathbb{I}_{[0,L-1]}$:
    \begin{equation}\label{prediction error eq}
        \begin{aligned}
        &\|\Bar{y}^*_k(t)-\hat{y}_{t+k}\|^2 \leq 3c_N\Bar{\varepsilon}^2\|\Bar{\alpha}^*(t)\|^2 + 3\|\Bar{\sigma}^*_k(t)\|^2\\
        &+3\rho_{2,l+k}(3l\Bar{\varepsilon}^2+3\|\Bar{\sigma}^*_{[-l,-1]}(t)\|^2+
        3lc_N\Bar{\varepsilon}^2\|\Bar{\alpha}^*(t)\|^2),
        \end{aligned}
    \end{equation}
    where $c_N \triangleq N-L-l+1$.
\end{mylem}
\textit{Proof}: The proof follows the same underlying argument as \cite[Lemma 2]{berberich2020data}, but is included in Appendix~\ref{proof of prediction error lemma} for completeness since the present setting employs an \(l\)-step extended state and Euclidean measurement noise bounds.

\cref{prediction error lemma} will later be used to construct a feasible solution for the MPC problem at the next time step, which is crucial for proving the recursive feasibility of the MPC scheme and practical exponential stability of the closed-loop system.
\subsection{Continuity Analysis of the Optimal Solution}\label{Analysis of the Optimal Solution Continuity}
In this section, we will introduce another technical result on the continuity of the optimal solution. Specifically, we aim to study the difference between the optimal solution/cost of Problem \eqref{robust data-driven mpc scheme} and the optimal solution/cost of the nominal MPC problem, which is provided as follows:
\begin{subequations}\label{nominal data-driven mpc scheme}
	\begin{align}
		&\check{J}^*_L(u_{[t-l,t-1]},y_{[t-l,t-1]})\nonumber \\ 
		&\begin{aligned}&= \min_{\substack{\check{\alpha}(t), \check{u}(t),\check{y}(t)}} \sum_{k=0}^{L-1}\left(\|\check{u}_k(t)-u^s\|_R^2 + \|\check{y}_k(t)-y^s\|_Q^2\right) \\
		&+ \|\check{\xi}_L(t)-\xi^s\|_P^2  + \lambda_\alpha \Bar{\varepsilon}^{\beta_\alpha}\|\check{\alpha}(t)\|^2\end{aligned}\label{nominal objective function} \\
        &\qquad \mathrm{s.t.}\quad \begin{bmatrix}
            \check{u}(t)\\
            \check{y}(t)
        \end{bmatrix} = \begin{bmatrix}
            H_{L+l}(u^d)\\
            H_{L+l}(y^d)
        \end{bmatrix}\check{\alpha}(t), \label{nominal parametrization using offline data}\\
        &\qquad\qquad\ \begin{bmatrix}
            \check{u}_{[-l,-1]}(t)\\
            \check{y}_{[-l,-1]}(t)
        \end{bmatrix} = \begin{bmatrix}
            u_{[t-l,t-1]}\\
            y_{[t-l,t-1]}
        \end{bmatrix}, \\
        &\qquad\qquad\ \check{\xi}_L(t)\in \Sigma_f, \check{\xi}_L(t) = \begin{bmatrix}
            \check{u}_{[L-l,L-1]}(t)\\
            \check{y}_{[L-l,L-1]}(t)
        \end{bmatrix}, \\
        &\qquad\qquad\ \check{u}_k(t)\in\mathbb{U}, k\in\mathbb{I}_{[0,L-1]}. 
	\end{align}
\end{subequations}
We denote its optimal solution by $\check{\alpha}^*(t), \check{u}^*(t),\check{y}^*(t)$. Note that this nominal MPC problem is not subject to measurement noise, which is different from our problem setting. Therefore, this problem is virtual, as the true noise-free online and offline outputs are inaccessible. We only use it for theoretical analysis and do not discuss its physical meaning. To carry out the analysis, we introduce another auxiliary MPC problem, which is presented below:
\begin{subequations}\label{auxiliary nominal data-driven mpc scheme}
	\begin{align}
		&\breve{J}^*_L(u_{[t-l,t-1]},y_{[t-l,t-1]})\nonumber \\ 
		&\begin{aligned}&= \min_{\substack{\breve{\alpha}(t), \breve{u}(t),\breve{y}(t)}} \sum_{k=0}^{L-1}\left(\|\breve{u}_k(t)-u^s\|_R^2 + \|\breve{y}_k(t)-y^s\|_Q^2\right) \\
		&+ \|\breve{\xi}_L(t)-\xi^s\|_P^2  + \lambda_\alpha \Bar{\varepsilon}^{\beta_\alpha}\|\breve{\alpha}(t)\|^2\end{aligned} \\
        &\qquad \mathrm{s.t.}\quad \begin{bmatrix}
            \breve{u}(t)\\
            \breve{y}(t) + \sigma_1
        \end{bmatrix} = \begin{bmatrix}
            H_{L+l}(u^d)\\
            H_{L+l}(y^d)
        \end{bmatrix}\breve{\alpha}(t), \\
        &\qquad\qquad\ \begin{bmatrix}
            \breve{u}_{[-l,-1]}(t)\\
            \breve{y}_{[-l,-1]}(t)
        \end{bmatrix} = \begin{bmatrix}
            u_{[t-l,t-1]}\\
            y_{[t-l,t-1]} + \sigma_2
        \end{bmatrix}, \\
        &\qquad\qquad\ \breve{\xi}_L(t)\in \Sigma_f, \breve{\xi}_L(t) = \begin{bmatrix}
            \breve{u}_{[L-l,L-1]}(t)\\
            \breve{y}_{[L-l,L-1]}(t)
        \end{bmatrix}, \\
        &\qquad\qquad\ \breve{u}_k(t)\in\mathbb{U}, k\in\mathbb{I}_{[0,L-1]}. 
	\end{align}
\end{subequations}
We denote its optimal solution by $\breve{\alpha}^*(t), \breve{u}^*(t),\breve{y}^*(t)$. In addition, we define $\sigma \triangleq \left[\begin{smallmatrix}
    \sigma_1\\
    \sigma_2
\end{smallmatrix}\right]$. Obviously, the only difference between Problem \eqref{auxiliary nominal data-driven mpc scheme} and Problem \eqref{nominal data-driven mpc scheme} is the presence of $\sigma$. This variable is pre-specified, and its expression will be provided in the proof of \cref{optimal solution continuity proposition}. It is evident that Problem \eqref{auxiliary nominal data-driven mpc scheme} is equivalent to Problem \eqref{nominal data-driven mpc scheme} when $\sigma = 0$. Before presenting the technical result, we refer to the formulation obtained by eliminating $\check{\alpha}(t)$ via the minimum-norm solution of the data equation as the reduced formulation of Problem \eqref{nominal data-driven mpc scheme}. A similar reduction procedure for Problem \eqref{robust data-driven mpc scheme} is provided in the proof of \cref{optimal solution continuity proposition}.
\begin{myass}\label{LICQ assumption}
The reduced formulation of Problem \eqref{nominal data-driven mpc scheme} satisfies the linear independence constraint qualification (LICQ) at its optimizer for all initial extended states sufficiently close to the equilibrium and all sufficiently small nonnegative noise bounds, i.e., the gradients of the equality constraints together with those of the active inequality constraints are linearly independent.
\end{myass}

To ensure that \cref{LICQ assumption} also covers the case $r=0$, an equivalent constraint representation is adopted. For $r=0$, the terminal inequality is replaced by an equivalent nonredundant linear terminal equality representation. We are now ready to provide the technical results below.
\begin{mypro}\label{optimal solution continuity proposition}
Suppose that Assumptions~\ref{persistently exciting assumption}--\ref{LICQ assumption} hold, that Problem \eqref{robust terminal set design sdp} is feasible, and that $\beta_\alpha>0, \beta_\sigma>0, \beta_\alpha+\beta_\sigma<2$ hold. Then, there exists a constant $\delta>0$ such that for all $\|\xi_t\|\leq \delta$, the following results hold:
\begin{enumerate}
    \item 
    \begin{subequations}\label{optimal solution continuity result 1}
    \begin{align}
        \check{J}^*_L(\xi_t) &\leq \big(c_{s,\xi}+ \lambda_\alpha\Bar{\varepsilon}^{\beta_\alpha} c_{\alpha,\xi}\big)\delta^2,\label{optimal solution continuity result 1.1}\\
        \Bar{J}^*_L(\Tilde{\xi}_t) &\leq \left(1+2c_H\tfrac{\lambda_\sigma}{\lambda_\alpha}\Bar{\varepsilon}^{2-\beta_\alpha-\beta_\sigma}\right)\check{J}^*_L(\xi_t)+2\lambda_\sigma l \Bar{\varepsilon}^{2-\beta_\sigma}, \label{optimal solution continuity result 1.2}
    \end{align}
    \end{subequations}
    where $c_H \triangleq (l+L)(N-L-l+1)$. $c_{s,\xi}$ and $c_{\alpha,\xi}$ are two constants that are defined in the proof of \cref{lower and upper bound lemma}.
    \item There exist functions $\beta_u \in \mathcal{K}_\infty$, $\beta_y \in \mathcal{K}_\infty$, and a constant $\bar{\varepsilon}_{\max}>0$ such that for all $\bar{\varepsilon}\in (0,\bar{\varepsilon}_{\max}]$,
    \begin{equation}\label{optimal solution continuity result 2}
        \begin{aligned}
            &\|\Bar{u}^*(t)-\check{u}^*(t)\|\leq \beta_u(\Bar{\varepsilon}),\\
            &\|\Bar{y}^*(t)-\check{y}^*(t)\|\leq \beta_y(\Bar{\varepsilon}).
        \end{aligned}
    \end{equation}
\end{enumerate}
\end{mypro}
\textit{Proof}: See Appendix~\ref{proof of optimal solution continuity proposition}.

\cref{optimal solution continuity proposition} provides a key bound on the difference between the optimal solution/cost of Problem \eqref{robust data-driven mpc scheme} and the optimal solution/cost of Problem \eqref{nominal data-driven mpc scheme}. Specifically, the result \eqref{optimal solution continuity result 1} implies that Problem \eqref{nominal data-driven mpc scheme} is feasible if $\|\xi_t\|$ is sufficiently small, and Problem \eqref{robust data-driven mpc scheme} is feasible whenever Problem \eqref{nominal data-driven mpc scheme} is feasible. Besides, $\Bar{J}^*_L(\Tilde{\xi}_t)$ is asymptotically upper-bounded by $\check{J}^*_L(\xi_t)$ as $\Bar{\varepsilon}\rightarrow 0$. The result \eqref{optimal solution continuity result 2} shows that if Problem \eqref{nominal data-driven mpc scheme} is feasible, the difference between the optimal solution of Problem \eqref{robust data-driven mpc scheme} and that of Problem \eqref{nominal data-driven mpc scheme} vanishes as $\bar{\varepsilon}\rightarrow0$. \cref{optimal solution continuity proposition} will be used to provide a key bound in the proof of practical exponential stability.

\subsection{Recursive Feasibility}\label{Recursive Feasibility}
\begin{mypro}\label{recursive feasibility proposition}
    Suppose that \cref{persistently exciting assumption} and \cref{terminal conditions assumption} hold and that the SDP \eqref{robust terminal set design sdp} is feasible. If the MPC problem \eqref{robust data-driven mpc scheme} is feasible at time step $t$, then it remains feasible at time step $t+1$.
\end{mypro}
\textit{Proof}: See Appendix~\ref{proof of recursive feasibility proposition}. 

It is worth noting that feasibility itself is not an issue for Problem \eqref{robust data-driven mpc scheme}. Under the persistent-excitation assumption, \(H_{L+l}(u^d)\) has full row rank, so any admissible input sequence satisfying the initialization and input constraints can be represented by some \(\bar{\alpha}(t)\). For any corresponding output sequence satisfying the initialization and terminal constraints, a slack variable \(\bar{\sigma}(t)\) can then be chosen to satisfy the data equation. Therefore, the purpose of \cref{recursive feasibility proposition} is not to establish feasibility, but to construct a feasible solution at time \(t+1\) directly from the optimal solution at time \(t\). This construction is essential for comparing successive optimal costs and forms the basis of the subsequent stability analysis.

\begin{myrem}
Unlike the construction in \cite{berberich2020data}, which appends a local steering segment to recover the terminal equality constraint, we simultaneously shift the optimal input and output trajectories and append the terminal feedback. Terminal set invariance then directly preserves the terminal constraint, providing the candidate for the one-step Lyapunov analysis.
\end{myrem}

\subsection{Practical Exponential Stability}\label{Practical Exponential Stability}
In this section, we provide the main theoretical result of this paper, which establishes the practical exponential stability of the closed-loop system under the proposed MPC scheme. This result is presented in the following theorem.

\begin{mythr}\label{practical exponential stability theorem}
Suppose that Assumptions~\ref{persistently exciting assumption}--\ref{LICQ assumption} hold and that Problem \eqref{robust terminal set design sdp} is feasible. Let
\[
\beta_\alpha>0,\qquad
\beta_\sigma>0,\qquad
\beta_\alpha+\beta_\sigma<2.
\]
Then there exist constants $\gamma>0, \delta>0$, 
$\Bar{\varepsilon}_{\max}>0$, and a function
$\beta\in\mathcal{K}_\infty$ such that, with
\[
\Bar V
:=
\gamma\lambda_{\min}(P_W)\delta^2,
\]
the noise bound $\Bar{\varepsilon}_{\max}$ can be chosen to satisfy
\[
\beta(\Bar{\varepsilon}_{\max})
\le \Bar V, c_{V,2}(\bar{\varepsilon}_{\max}) < \Bar V.
\]
For any $\Bar{\varepsilon}\in
(0,\Bar{\varepsilon}_{\max}]$, any admissible noise sequence satisfying
this bound, and any initial condition such that
\[
V_0\le \Bar V,
\]
the one-step implementation of
\eqref{robust data-driven mpc scheme} is recursively feasible. Moreover, the sublevel set $V_t\leq \bar{V}$ is invariant. Furthermore, the closed-loop extended state is practically
exponentially stable. Consequently, the system state converges to a ball around the origin whose radius vanishes as
$\Bar{\varepsilon}\to0$.
\end{mythr}
\textit{Proof}: See Appendix~\ref{proof of practical exponential stability theorem}. 

\cref{practical exponential stability theorem} is the main theoretical result of this paper, which establishes the practical exponential stability of the closed-loop system under the proposed MPC scheme. The noise-free case, i.e., $\bar{\varepsilon}=0$, reduces to the formulation considered in \cite{berberich2021design}, which proved exponential stability. In addition to the formal assumptions stated in this paper, we also assume that Problem \eqref{robust terminal set design sdp} is feasible, which enables us to design the weight matrix $P$ and terminal set $\Sigma_f$ satisfying \cref{terminal conditions assumption}. Besides, the conditions $\beta_\alpha>0, \beta_\sigma>0, \beta_\alpha+\beta_\sigma<2$ are required in the proof, ensuring that the noise term vanishes as $\Bar{\varepsilon}\rightarrow 0$. Furthermore, according to \eqref{lower and upper bound eq}, the choice $\Bar{V} = \gamma\lambda_{\min}(P_W)\delta^2$ implies that $\|x_t\|\leq\|\Gamma\|\|\xi_t\|\leq\|\Gamma\|\delta$, i.e., our result holds locally. Furthermore, we require the online and offline measurement noise to be sufficiently small; the admissible noise bound \(\bar\varepsilon_{\max}\) must be chosen sufficiently small relative to \(\bar V\).

\subsection{Discussion and Comparison with Existing Data-Driven MPC Schemes}\label{remark on different mpc scheme}
    The proposed formulation unifies existing terminal-constraint data-driven MPC schemes. Specifically, the terminal equality formulation in \cite{berberich2020data} is recovered when $r=0$, whereas the noise-free terminal set formulation in \cite{berberich2021design} is recovered when $\Bar{\varepsilon}=0$. Beyond this theoretical unification, the introduction of a terminal set also provides additional flexibility in closed-loop performance. Compared with the terminal equality constraints adopted in \cite{berberich2020data,liu2022data}, a nonzero terminal set level $r$ relaxes the terminal requirement and can avoid overly aggressive control actions, thereby yielding smoother control inputs and improved closed-loop performance. Hence, $r$ can be interpreted as a performance-tuning parameter that provides a trade-off between the restrictiveness of the terminal condition and the resulting closed-loop behavior. This effect will be demonstrated numerically in \cref{level study}.

Second, our formulation avoids the nonconvex constraint $\|\sigma(t)\|_\infty\leq \Bar{\varepsilon}(1+\|\alpha(t)\|_1)$ employed in \cite{berberich2020data} for the theoretical analysis. Instead, the noise-dependent regularization term $\frac{\lambda_\sigma}{\Bar{\varepsilon}^{\beta_\sigma}}\|\Bar{\sigma}(t)\|^2$ increasingly penalizes the slack variable as $\Bar{\varepsilon}$ decreases and thus drives $\Bar{\sigma}(t)$ toward zero in the small-noise case. The same regularization mechanism could also be incorporated into \cite{berberich2020data} without changing its stability argument.

Finally, compared with data-driven MPC without terminal ingredients \cite{bongard2022robust}, our framework does not rely on a sufficiently long prediction horizon to establish closed-loop stability and instead only requires the lower bound $L\geq l+n_x$. This allows a substantially shorter prediction horizon to be employed and can therefore reduce the online computational burden. The price is an additional offline step for designing the terminal ingredients. Such a trade-off between offline and online computation is particularly attractive in networked and resource-constrained control systems, where substantial computational resources may be available offline while online computation is subject to real-time limitations. The effect of the terminal set is evaluated numerically in \cref{terminal study}.

\section{Simulations}\label{Simulations}
In this section, numerical simulations are conducted to demonstrate the effectiveness and advantages of the proposed robust data-driven MPC scheme. Specifically, we investigate three aspects. First, we study how the size of the terminal set affects the closed-loop behavior and demonstrate the tunable trade-off enabled by the proposed terminal set formulation. Second, we investigate the influence of the prediction horizon and show that terminal ingredients are particularly beneficial when a short prediction horizon is employed. Third, we compare the proposed method with data-driven MPC without terminal ingredients and demonstrate that comparable closed-loop performance can be achieved with a substantially shorter prediction horizon and reduced online computation.

For all simulations, we consider system \eqref{basic system} with 
\begin{equation*}
    A=
\begin{bmatrix}
1.0121 & 0.1445\\
0.1084 & 0.5061
\end{bmatrix},
B=
\begin{bmatrix}
0.0225\\
0.2024
\end{bmatrix}, C = \begin{bmatrix}
1 &0
\end{bmatrix}, D = 0.
\end{equation*}
We can easily compute the observability index as $\underline{l}=2$. Consequently, we choose $l=2$. The above system is obtained by zero-order-hold (ZOH) discretization of the continuous-time dynamics
\begin{equation*}
    \dot x=
\begin{bmatrix}
0&1\\
0.75&-3.5
\end{bmatrix}x+
\begin{bmatrix}
0\\
1.4
\end{bmatrix}u
\end{equation*}
with sampling time \(T_s=0.2\,\mathrm{s}\). The system is open-loop unstable and can be interpreted as a second-order benchmark motivated by linearized inverted pendulum dynamics around its upright equilibrium. The control objective is to regulate the system to the equilibrium $(u^s,y^s)=(0,0)$ while satisfying the input constraint $u_t\in\mathbb U=[-4,4]$. Since only the first state, i.e., angular position, is measured, the controller is implemented in an output-feedback fashion using only measured input-output data. The length of the offline input-output data trajectory is set to $N=200$, and the online and offline measurement noise bounds are chosen as $\bar{\varepsilon}_{on} = 
\bar{\varepsilon}_{off}=0.001$. 
The weighting and regularization parameters are selected as
\begin{equation*}
Q=1, R=2,
\beta_{\alpha}=0.5,
\beta_{\sigma}=0.5, \lambda_{\alpha}=100,
\lambda_{\sigma}=4000.
\end{equation*}
Using the collected offline data, we first solve SDP \eqref{robust terminal set design sdp} in \cref{robust terminal set design theorem} to construct the terminal ingredients. The resulting terminal weighting matrix $P$ and a conservative reference level $r_{ref}$ are given by
\begin{equation*}
    P = \left[\begin{smallmatrix}
        2.4262	&0.7569	&-6.6579	&15.5241\\
0.7569	&5.5529	&-16.4195	&34.8040\\
-6.6579	&-16.4195	&149.9706	&-276.9169\\
15.5241	&34.8040	&-276.9169	&550.0696
    \end{smallmatrix}\right], r_{ref} = 7.6328.
\end{equation*}
All numerical experiments are performed in MATLAB R2024b using MOSEK on a computer equipped with an Intel Core i7-12700H CPU and 16 GB of RAM.

To facilitate further discussion, we define some important metrics: realized cost $J_{real}$, total input variation $TV_u$, and peak input deviation $PD_u$:
\begin{equation*}
    \begin{aligned}
        J_{real} &= \sum_{k=0}^{N_{sim}-1}\left(\|u_k\|^2_R+\|y_k\|^2_Q\right),\\
        TV_u &= \sum_{k=0}^{N_{sim}-1}
\|u_{k}-u_{k-1}\|, PD_u = \max_{k=0,\ldots,N_{sim}-1}\|u_k\|,
    \end{aligned}
\end{equation*}
where $N_{sim}$ represents the simulation length.

\subsection{Terminal Set Size and Performance Trade-off}\label{level study}
We first investigate the influence of the terminal set size on the closed-loop behavior. Recall the expression of the terminal set \eqref{expression of terminal set}. When $r=0$, the terminal set reduces to the equilibrium point, and the proposed MPC problem recovers the terminal equality formulation considered in \cite{berberich2020data}. As $r$ increases, the terminal condition is gradually relaxed from a point constraint to a nontrivial terminal region. To characterize the size of the terminal set independently of its absolute scale, we use the normalized terminal set level $r/r_{ref}$.

\cref{B_level_tradeoff_publication_trajectories} illustrates the closed-loop trajectories for several representative terminal set levels, where we choose $N_{sim}=51$. We use TE-$1$ to denote the one-step implementation of \eqref{robust data-driven mpc scheme} with a terminal equality constraint. Similarly, we use TS-$1$ to represent the one-step implementation of \eqref{robust data-driven mpc scheme} with a terminal set constraint, and the closed-loop trajectories for several representative normalized levels are plotted. TE-$n_x$ is the baseline method proposed in \cite{berberich2020data}. For the TE-$1$ and TE-$n_x$ cases, the controller forces the predicted trajectory to reach the equilibrium exactly at the end of the prediction horizon, which leads to fast output convergence but relatively aggressive control actions. As $r/r_{ref}$ increases, the terminal condition becomes progressively less restrictive. The controller is therefore allowed to approach the equilibrium more gradually, resulting in slower output convergence but smoother and less aggressive input trajectories. In particular, the peak control input decreases noticeably as $r/r_{ref}$ increases. We summarize the quantitative results in \cref{B_level_tradeoff_metrics}, which validates the above observation. It is worth noting from \cref{B_level_tradeoff_publication_trajectories} and \cref{B_level_tradeoff_metrics} that TE-$1$ and TE-$n_x$ exhibit similar transient behavior; TE-$1$ achieves lower realized cost, whereas TE-$n_x$ results in slightly lower total input variation. The rigorous comparison under the noisy setting is beyond the scope of this paper and is left for future research. 
\begin{figure}[t]
    \centering
    \hfill 
    \includegraphics[width=\linewidth]{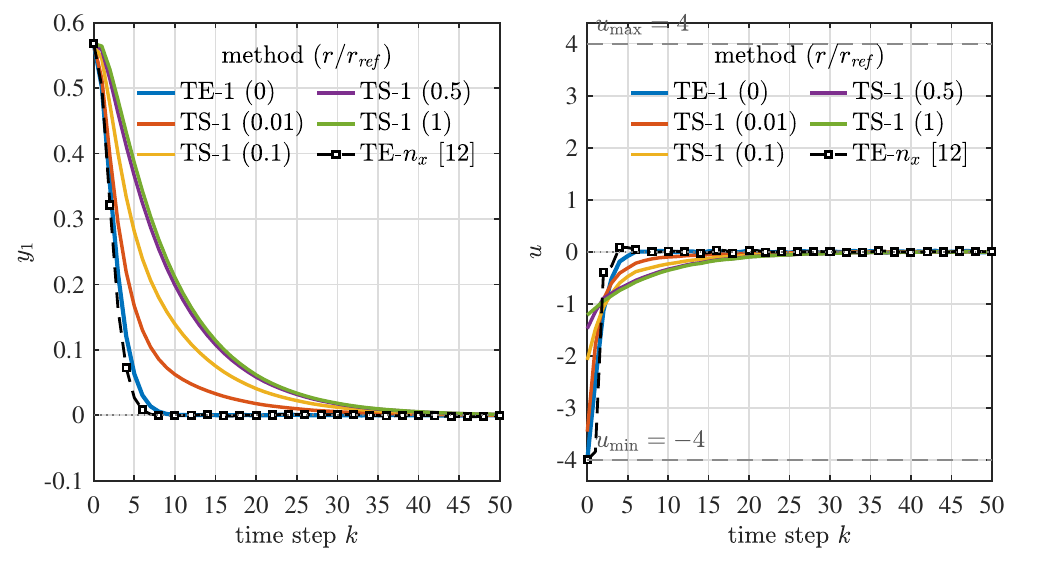}
    \caption{Closed-loop comparison of TE-$1$, TS-$1$, and the TE-\(n_x\) baseline under different terminal set levels.}
    \label{B_level_tradeoff_publication_trajectories}
\end{figure}
\begin{table}[!t]
\centering
\caption{Closed-loop performance metrics under different terminal set levels (with $L=6$).}
\label{B_level_tradeoff_metrics}
\begin{tabular}{c c c c}
\toprule
Design & $J_{real}$ & $TV_u$ & $PD_u$ \\
\midrule
TE-$1$ $(0)$ &$48.2439$  &$8.7550$  &$4.0000$   \\
TS-$1$ $(0.01)$ &$34.5222$  &$7.0126$  &$3.4564$  \\
TS-$1$ $(0.1)$ &$20.5470$  &$4.2349$  &$2.0780$  \\
TS-$1$ $(0.5)$ &$16.3367$  &$3.0211$  &$1.4770$  \\
TS-$1$ $(1)$ &$15.6649$  &$2.4966$  &$1.2157$ \\
TE-$n_x$ \cite{berberich2020data} &$62.7684$  &$8.5650$  &$4.0000$ \\
\bottomrule
\end{tabular}
\end{table}
\begin{table}[!t]
    \caption{Closed-loop performance comparison of terminal equality 
    and terminal set across prediction horizons.}
    \label{C_performance_table}
    \centering
    \renewcommand{\arraystretch}{1.08}

    \begin{tabular}{@{}ccrrrr@{}}
      \toprule
      $r/r_{ref}$ & Metric
        & $L=6$ & $L=10$ & $L=15$ & $L=20$ \\
      \midrule

      \multirow{3}{*}{$0$}
        & $J_{real}$
        & $48.2439$ & $20.1970$ & $15.3484$ & $14.3964$ \\
        & $TV_u$
        & $8.7550$ & $3.7563$ & $2.3842$ & $1.9577$ \\
        & $PD_u$
        & $4.0000$ & $1.7947$ & $1.1645$ & $0.9715$ \\
      \midrule

      \multirow{3}{*}{$1$}
        & $J_{real}$
        & $15.6649$ & $14.5990$ & $14.2173$ & $14.1196$ \\
        & $TV_u$
        & $2.4966$ & $2.0683$ & $1.8303$ & $1.7224$ \\
        & $PD_u$
        & $1.2157$ & $1.0215$ & $0.9144$ & $0.8661$ \\
      \bottomrule
    \end{tabular}
\end{table}

\subsection{Effect of the Prediction Horizon}\label{horizon study}
We next investigate the influence of the prediction horizon $L$. \cref{C_performance_table} compares the closed-loop performance of the terminal equality case $r=0$ and the largest terminal set $r/r_{ref}=1$ under different prediction horizons. The simulation length $N_{sim}$ is chosen as $51$, consistent with the previous experiment. For the short horizon \(L=6\), the terminal set formulation provides a substantial improvement: \(J_{real}\) decreases from \(48.2439\) to \(15.6649\), while \(TV_u\) and \(PD_u\) decrease from \(8.7550\) and \(4.0000\) to \(2.4966\) and \(1.2157\), respectively. This shows that relaxing the terminal equality constraint is particularly beneficial when the prediction horizon is short.

As \(L\) increases, the performance gap between the two terminal conditions becomes progressively smaller. At \(L=20\), the corresponding values are already very close, with \(J_{real}=14.3964\) and \(14.1196\) for \(r=0\) and \(r/r_{ref}=1\), respectively. Similar trends can be observed for \(TV_u\) and \(PD_u\). We can provide a qualitative analysis for the above observations. When $L$ is small, a controller with the terminal equality constraint is required to steer the predicted trajectory to a prescribed terminal point within only a few steps. This restrictive requirement may result in relatively large and aggressive control actions. In contrast, the terminal set constraint only requires the terminal trajectory to enter a prescribed set, thereby providing additional freedom to the optimizer and potentially yielding smoother control inputs and improved closed-loop performance. Hence, a more significant performance difference between terminal equality and terminal set constraints is expected in the short-horizon case. As $L$ increases, the controller has a longer prediction window to approach the terminal condition gradually, which naturally mitigates aggressive control actions even under a terminal equality constraint. Consequently, the performance difference between the two terminal formulations is expected to vanish for sufficiently long prediction horizons.

\begin{figure}[t]
    \centering
    \hfill 
    \includegraphics[width=\linewidth]{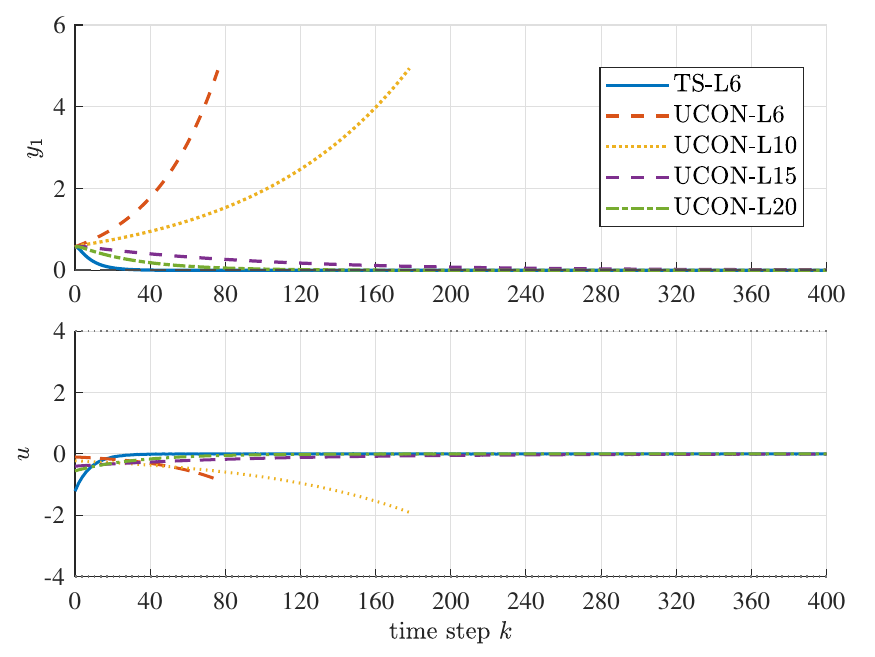}
    \caption{Closed-loop comparison between TS-L$6$ and UCON under different prediction horizons.}
    \label{D_ucon_initial_transient}
\end{figure}
\begin{table}[!t]
\centering
\caption{Comparison of closed-loop performance metrics for TS-L$6$ and UCON with different prediction horizons.}
\label{D_ucon_metrics}
\begin{tabular}{c c c c c}
\toprule
Design & $J_{real}$ & $TV_u$ & $PD_u$ &$T_{tot}$ \\
\midrule
TS-L$6$  &$15.6697$  &$3.7409$  &$1.2157$ &\(129.4502\,\mathrm{s}\) \\
UCON-L$6$  &$\times$  &$\times$ &$\times$ &$\times$\\
UCON-L$10$  &$\times$  &$\times$  &$\times$ &$\times$ \\
UCON-L$15$ &$33.7503$ &$1.0077$  &$0.3998$ & \(172.3504\,\mathrm{s}\)\\
UCON-L$20$  &$16.8952$ &$1.5441$  &$0.5522$ & \(199.8267\,\mathrm{s}\)\\ 
\bottomrule
\end{tabular}
\vspace{2pt}

    \begin{minipage}{0.95\columnwidth}
        Here, \footnotesize{\(\times\) denotes closed-loop divergence, for which the corresponding performance metrics are not reported.}
    \end{minipage}
\end{table}
\subsection{Comparison with MPC Without Terminal Ingredients}\label{terminal study}
We finally compare the proposed scheme with the data-driven MPC framework proposed in \cite{bongard2022robust}, denoted by UCON. UCON does not employ terminal ingredients and can be viewed as the counterpart of our framework obtained by removing the terminal cost and terminal set constraint. In contrast to the proposed approach, which uses terminal ingredients to guarantee desirable closed-loop behavior with a relatively short prediction horizon, UCON relies on a sufficiently long prediction horizon to achieve stability and satisfactory performance, which is the main theoretical result of \cite{bongard2022robust}. This distinction is also consistent with the theoretical motivation of our framework.

To investigate this difference, we compare TS-L$6$ with UCON using prediction horizons \(L=6,10,15,\) and \(20\). Here, the number following “L” denotes the prediction horizon. We set the simulation length to $N_{sim}=401$. \cref{D_ucon_initial_transient} shows a clear dependence of UCON on the prediction horizon. UCON-L$6$ and UCON-L$10$ fail to stabilize the considered unstable system. Increasing the horizon to \(L=15\) restores convergence, but the resulting closed-loop performance remains substantially worse than that of TS-L$6$. Only when the prediction horizon is further increased to \(L=20\) does UCON achieve closed-loop performance comparable to TS-L$6$. In contrast, the proposed controller already achieves satisfactory stabilization and performance with the much shorter horizon \(L=6\). The comparable performance of TS-L$6$ and UCON-L$20$ highlights the benefit of terminal ingredients, which enable similar closed-loop performance with a much shorter prediction horizon and substantially lower online computation. The quantitative results are summarized in \cref{D_ucon_metrics}, where $T_{tot}$ denotes the total online implementation time. These results are consistent with the discussion in \cref{remark on different mpc scheme}, confirming that the proposed framework trades additional offline computation for reduced online computational effort without sacrificing closed-loop performance.
\section{Conclusion}\label{Conclusion}
This paper develops a unified robust data-driven MPC framework with terminal ingredients for unknown LTI systems under noise-free input and noisy output measurements. Robust terminal ingredients are designed directly from offline data, and a one-step MPC implementation is established. The proposed scheme guarantees recursive feasibility and practical exponential stability under suitable conditions. Numerical results demonstrate that the terminal set level provides a useful trade-off between convergence speed and closed-loop performance. In particular, the proposed terminal set formulation is advantageous for short prediction horizons and can achieve performance comparable to long-horizon terminal-free MPC with reduced online computational effort.
\appendix
\subsection{Proof of \cref{lower and upper bound lemma}}\label{proof of lower and upper bound lemma}
The lower bound is trivial. We only focus on the upper bound. According to \eqref{relationship between two systems 2}, we have $\|x_t\| \leq \|\Gamma\| \|\xi_t\|$. Therefore, if $\delta$ is sufficiently small, then $\xi_t$, and hence $x_t$, will also be sufficiently small. Moreover, \cref{prediction horizon assumption} implies $L-l\geq n_x$. These two facts imply that there exists a feasible input sequence $\Hat{u}_{[t,t+L-l-1]}$, i.e., $\Hat{u}_k\in\mathbb{U}, k\in\mathbb{I}_{[t,t+L-l-1]}$, that steers $x_t$ to $x_{t+L-l} = 0$ while satisfying 
\begin{equation}\label{uy bounded by x}
    \left\lVert
    \begin{bmatrix}
        \Hat{u}_{[t,t+L-l-1]}\\
        \Hat{y}_{[t,t+L-l-1]}
    \end{bmatrix}
    \right\rVert^2 \leq \Gamma_{uy} \|x_t\|^2
\end{equation}
for some constant $\Gamma_{uy}>0$. This fact can be explained by choosing $\Hat{u}$ with minimal energy and leveraging the system dynamics. Then, we can construct a candidate input/output solution for Problem \eqref{robust data-driven mpc scheme} as $\Bar{u}_{[0,L-l-1]}(t) = \Hat{u}_{[t,t+L-l-1]}$, $\Bar{u}_{[L-l,L-1]}(t)=0$, $\Bar{y}_{[0,L-l-1]}(t) = \Hat{y}_{[t,t+L-l-1]}$, $\Bar{y}_{[L-l,L-1]}(t)=0$. This solution also satisfies the terminal constraint since the terminal set contains the origin. To facilitate subsequent discussion, we define 
\begin{equation}
    H_{ux} \triangleq
    \begin{bmatrix}
        H_{L+l}(u^d)\\
        H_1\left(x^d_{[0,N-L-l]}\right)
    \end{bmatrix},
\end{equation}
which has full row rank according to \cref{persistently exciting assumption} \cite[Corollary 2]{willems2005note}. By using $H_{ux}$, we choose the candidate $\Bar{\alpha}(t)$ as
\begin{equation}\label{alpha bar expression}
    \Bar{\alpha}(t) = H_{ux}^\dagger
    \begin{bmatrix}
        \Bar{u}_{[-l,L-1]}(t)\\
        x_{t-l}
    \end{bmatrix}.
\end{equation}
Note that this $\Bar{\alpha}(t)$ implies that 
\begin{equation}
    \begin{bmatrix}
        H_{L+l}(u^d)\\
        H_{L+l}(y^d)
    \end{bmatrix}\Bar{\alpha}(t)=
    \begin{bmatrix}
        \Bar{u}_{[-l,L-1]}(t)\\
        y_{[t-l,t-1]}\\
        \Bar{y}_{[0,L-1]}(t)
    \end{bmatrix}.
\end{equation}
Therefore, to ensure that the equality constraint \eqref{parametrization using offline data} holds, we choose $\Bar{\sigma}(t)$ as
\begin{equation}
    \begin{aligned}
        \Bar{\sigma}_{[-l,-1]}(t) &= H_l\left(\varepsilon^d_{[0,N-L-1]}\right)\Bar{\alpha}(t)-\varepsilon_{[t-l,t-1]},\\
        \Bar{\sigma}_{[0,L-1]}(t) &= H_L\left(\varepsilon^d_{[l,N-1]}\right)\Bar{\alpha}(t).
    \end{aligned}
\end{equation}
We have shown that our candidate solution is feasible for Problem \eqref{robust data-driven mpc scheme}. In addition, it is worth noting that this candidate can provide an upper bound on $V_t$ since $V_t$ is defined as the optimal cost of Problem \eqref{robust data-driven mpc scheme} plus $\gamma W(\xi_t)$; see \cref{Lower and Upper bound of Lyapunov Function}. Our next step is thus to bound the cost of this candidate solution. First,
\begin{equation}\label{sigma bound}
    \begin{aligned}
    \|\Bar{\sigma}(t)\|^2 &\leq 2\|H_{L+l}(\varepsilon^d_{[0,N-1]})\Bar{\alpha}(t)\|^2 + 2\|\varepsilon_{[t-l,t-1]}\|^2\\
    &\leq 2 \|H_{L+l}(\varepsilon^d_{[0,N-1]})\|^2 \|\Bar{\alpha}(t)\|^2 + 2l\Bar{\varepsilon}^2\\
    &\leq 2c_H\Bar{\varepsilon}^2\|\Bar{\alpha}(t)\|^2 + 2l\Bar{\varepsilon}^2,
    \end{aligned}
\end{equation}
where we use the fact that the matrix $2$-norm is bounded by the Frobenius norm in the third inequality. We then bound $\|\Bar{\alpha}(t)\|^2$ as follows:
\begin{equation}
    \begin{aligned}
        &\|\Bar{\alpha}(t)\|^2 \leq \|H_{ux}^\dagger\|^2\left(\|\Bar{u}_{[-l,L-1]}(t)\|^2+\|x_{t-l}\|^2\right)\\
        &= \|H_{ux}^\dagger\|^2\left(\|\Bar{u}_{[0,L-1]}(t)\|^2+\|\Bar{u}_{[-l,-1]}(t)\|^2 + \|x_{t-l}\|^2\right).
    \end{aligned}
\end{equation}
Notice that $\left[\begin{smallmatrix}\Bar{u}_{[-l,-1]}(t)\\x_{t-l}\end{smallmatrix}\right] = \left[\begin{smallmatrix}I &0\\-\mathcal{O}_l^\dagger\mathcal{T}_l &\mathcal{O}_l^\dagger\end{smallmatrix}\right]\xi_t$, $\|\Bar{u}_{[0,L-1]}(t)\|^2\leq \Gamma_{uy}\|x_t\|^2$. By using these facts and \eqref{relationship between two systems 2}, we have
\begin{equation}\label{alpha bound}
\|\Bar{\alpha}(t)\|^2 \leq \underbrace{\|H_{ux}^\dagger\|^2\left(\left\lVert\left[\begin{smallmatrix}I &0\\-\mathcal{O}_l^\dagger\mathcal{T}_l &\mathcal{O}_l^\dagger\end{smallmatrix}\right]\right\rVert^2+\Gamma_{uy}\|\Gamma\|^2\right)}_{c_{\alpha,\xi}}\|\xi_t\|^2.
\end{equation}
In addition, \eqref{uy bounded by x} and \eqref{relationship between two systems 2} imply
\begin{equation}\label{sum bound}
    \sum_{k=0}^{L-1}\left(\|\Bar{u}_k(t)\|_R^2 + \|\Bar{y}_k(t)\|_Q^2\right)\leq \underbrace{\lambda_{\max}\left(Q,R\right)\Gamma_{uy}\|\Gamma\|^2}_{c_{s,\xi}} \|\xi_t\|^2.
\end{equation}
By substituting \eqref{sum bound}, \eqref{sigma bound}, and \eqref{alpha bound} into the cost function \eqref{objective function}, we obtain
\begin{equation}\label{proof of lower and upper bound case 1}
    V_t  \leq c_{V,1}(\bar{\varepsilon}) \|\xi_t\|^2 + c_{V,2}(\bar{\varepsilon}),
\end{equation}
where
\begin{equation}
\begin{aligned}   
    c_{V,1}(\bar{\varepsilon}) &= c_{s,\xi} + \left(\lambda_\alpha\Bar{\varepsilon}^{\beta_\alpha} + 2\lambda_\sigma c_H\Bar{\varepsilon}^{2-\beta_\sigma}\right)c_{\alpha,\xi}+ \gamma\lambda_{\max}(P_W),\\
    c_{V,2}(\bar{\varepsilon}) &= 2l\lambda_\sigma\Bar{\varepsilon}^{2-\beta_\sigma}.
\end{aligned}
\end{equation}
The proof is complete. $\hfill \square$
\subsection{Proof of \cref{prediction error lemma}}\label{proof of prediction error lemma}
From \eqref{parametrization using offline data} and \eqref{initialization}, we notice that 
\begin{equation}
    \begin{bmatrix}
            H_{L+l}(u^d)\\
            H_{L+l}(y^d)
        \end{bmatrix}\Bar{\alpha}^*(t)
\end{equation}
is a trajectory of \eqref{basic system}. Then,
\begin{equation}
    y^- \triangleq \hat{y} -  H_{L+l}(y^d)\Bar{\alpha}^*(t)
\end{equation}
represents a zero-input output trajectory with 
\begin{equation}\label{inital condition Proof of Lemma 4}
     y^-_{[t-l,t-1]} = H_l\left(\varepsilon^d_{[0,N-L-1]}\right)\bar{\alpha}^*(t)-\varepsilon_{[t-l,t-1]}-\bar{\sigma}^*_{[-l,-1]}(t).
\end{equation}
According to the system dynamics of \eqref{basic system}, for $k\in \mathbb{I}_{[0,L-1]}$,  
\begin{equation}\label{leverage system dynamics Proof of Lemma 4}
     y^-_{t+k} = CA^{l+k}\mathcal{O}_l^\dagger y^-_{[t-l,t-1]}.
\end{equation}
Therefore,
\begin{equation}\label{compute bound Proof of Lemma 4}
\begin{aligned}
   & \|\bar{y}^*_k(t)-\hat{y}_{t+k}\| \\
   &= \|H_1\left(\tilde{y}^d_{[l+k,N-L+k]}\right)\bar{\alpha}^*(t)-\bar{\sigma}^*_k(t)-\hat{y}_{t+k}\|\\
   &\leq \|y^-_{t+k}\| + \|H_1\left(\varepsilon^d_{[l+k,N-L+k]}\right)\bar{\alpha}^*(t)-\bar{\sigma}^*_k(t)\|.
\end{aligned}
\end{equation}
Combining \eqref{inital condition Proof of Lemma 4}, \eqref{leverage system dynamics Proof of Lemma 4}, and \eqref{compute bound Proof of Lemma 4}, and using the facts that 
\begin{equation}
    \|a+b+c\|^2\leq3\left(\|a\|^2+\|b\|^2+\|c\|^2\right), \|H_q(\varepsilon^d)\|^2\leq qc_N \bar{\varepsilon}^2,
\end{equation}
we obtain \eqref{prediction error eq}. The proof is complete. $\hfill \square$
\subsection{Proof of \cref{optimal solution continuity proposition}}\label{proof of optimal solution continuity proposition}
If $\delta>0$ is sufficiently small, then $\|\xi_t\|\leq\delta$ implies that Problem \eqref{nominal data-driven mpc scheme} is feasible. We can prove \eqref{optimal solution continuity result 1.1} by using the candidate solution constructed in the proof of \cref{lower and upper bound lemma}. To prove \eqref{optimal solution continuity result 1.2}, we construct a feasible candidate solution for Problem \eqref{robust data-driven mpc scheme} based on the optimal solution of Problem \eqref{nominal data-driven mpc scheme}. We choose $\Bar{u}(t)=\check{u}^*(t), \Bar{y}(t)=\left[\begin{smallmatrix}\Tilde{y}_{[t-l,t-1]}\\\check{y}^*_{[0,L-1]}(t)\end{smallmatrix}\right]$. Note that $\|\check{\alpha}(t)\|^2$ is penalized in the cost function \eqref{nominal objective function}, which implies that $\check{\alpha}^*(t)$ satisfying \eqref{nominal parametrization using offline data} has minimal norm. Since \eqref{nominal parametrization using offline data} parametrizes a system trajectory due to the absence of measurement noise, we can obtain the expression of $\check{\alpha}^*(t)$ as
\begin{equation}\label{alpha check expression}
    \check{\alpha}^*(t) = H_{ux}^\dagger \begin{bmatrix}\check{u}^*(t)\\x_{t-l}\end{bmatrix}.
\end{equation}
Then, we choose $\Bar{\alpha}(t) = \check{\alpha}^*(t)$. To ensure that the equality constraint \eqref{parametrization using offline data} holds, we choose $\Bar{\sigma}(t)$ as
\begin{equation}
    \Bar{\sigma}(t) = H_{L+l}(\Tilde{y}^d)\Bar{\alpha}(t)- \Bar{y}(t).
\end{equation}
It is easy to verify that the candidate solution $\Bar{\alpha}(t),\Bar{\sigma}(t),\Bar{u}(t),\Bar{y}(t)$ satisfies all constraints of Problem \eqref{robust data-driven mpc scheme} and is thus feasible. Now, we are ready to characterize the relation between the optimal cost $\Bar{J}^*_L(\Tilde{\xi}_t)$ of Problem \eqref{robust data-driven mpc scheme} and the optimal cost $\check{J}^*_L(\xi_t)$ of Problem \eqref{nominal data-driven mpc scheme}. First, by using the candidate solution $  \Bar{\alpha}(t),\Bar{\sigma}(t),\Bar{u}(t),\Bar{y}(t)$, we have
\begin{equation}\label{objective function bound proof of optimal solution continuity}
    \Bar{J}^*_L(\Tilde{\xi}_t) \leq \check{J}^*_L(\xi_t) + \frac{\lambda_\sigma}{\Bar{\varepsilon}^{\beta_\sigma}}\|\Bar{\sigma}(t)\|^2.
\end{equation}
Then, we bound $\|\Bar{\sigma}(t)\|^2$ as follows:
\begin{equation}\label{sigma bound proof of optimal solution continuity}
    \begin{aligned}
        &\|\Bar{\sigma}(t)\|^2 \\
        &= \left\lVert\left(H_{L+l}(\Tilde{y}^d)-H_{L+l}(y^d)\right)\Bar{\alpha}(t)+H_{L+l}(y^d)\Bar{\alpha}(t)-\Bar{y}(t)\right\rVert^2\\
        &\leq 2\|H_{L+l}(\varepsilon^d)\Bar{\alpha}(t)\|^2+2\|\Check{y}^*(t)-\Bar{y}(t)\|^2\\
        &\leq 2\|H_{L+l}(\varepsilon^d)\|^2\|\Bar{\alpha}(t)\|^2+2\|\varepsilon_{[t-l,t-1]}\|^2\\
        &\leq 2c_H\Bar{\varepsilon}^2\|\Bar{\alpha}(t)\|^2+2l\Bar{\varepsilon}^2.
    \end{aligned}
\end{equation}
The term $\|\Bar{\alpha}(t)\|^2$ can be bounded as follows:
\begin{equation}\label{alpha bound proof of optimal solution continuity}
    \|\Bar{\alpha}(t)\|^2 = \|\check{\alpha}^*(t)\|^2 \leq \frac{\check{J}^*_L(\xi_t)}{\lambda_\alpha \Bar{\varepsilon}^{\beta_\alpha}}.
\end{equation}
Combining \eqref{objective function bound proof of optimal solution continuity}, \eqref{sigma bound proof of optimal solution continuity}, and \eqref{alpha bound proof of optimal solution continuity}, we obtain \eqref{optimal solution continuity result 1.2}. 

For \eqref{optimal solution continuity result 2}, we will only prove $\|\Bar{u}^*(t)-\check{u}^*(t)\|\leq \beta_u(\Bar{\varepsilon})$ since $\|\Bar{y}^*(t)-\check{y}^*(t)\|\leq \beta_y(\Bar{\varepsilon})$ can be proved in a similar way. To prove this result, we need to consider Problem \eqref{auxiliary nominal data-driven mpc scheme}. The expression of $\sigma$ is given by
\begin{equation}
    \sigma = \begin{bmatrix}
        \sigma_1\\ \sigma_2
    \end{bmatrix} = 
    \begin{bmatrix}
        \Bar{\sigma}^*(t) + \left(H_{L+l}(y^d)-H_{L+l}(\Tilde{y}^d)\right)\Bar{\alpha}^*(t)\\
        \Tilde{y}_{[t-l,t-1]} - y_{[t-l,t-1]}
    \end{bmatrix} \triangleq \tilde{\sigma} .
\end{equation}
Due to the construction of $\sigma$, $\left(\Bar{\alpha}^*(t),\Bar{u}^*(t),\Bar{y}^*(t)\right)$ is a feasible but not necessarily optimal solution of Problem \eqref{auxiliary nominal data-driven mpc scheme}. The corresponding objective function value is denoted by $\breve{J}_L(\xi_t)$. It is obvious that the following relation holds:
\begin{equation}\label{relationship 1 Proof of Proposition 1}
    \breve{J}^*_L(\xi_t) \leq\breve{J}_L(\xi_t) = \Bar{J}^*_L(\Tilde{\xi}_t) - \frac{\lambda_\sigma}{\Bar{\varepsilon}^{\beta_\sigma}}\|\Bar{\sigma}^*(t)\|^2.
\end{equation}
We then aim to construct a feasible solution to Problem \eqref{robust data-driven mpc scheme} based on $\left(\breve{\alpha}^*(t),\breve{u}^*(t),\breve{y}^*(t)\right)$. By choosing
\begin{equation}
    \Bar{\sigma}(t) = \Bar{\sigma}^*(t) + \left(H_{L+l}(\Tilde{y}^d)-H_{L+l}(y^d)\right)\left(\breve{\alpha}^*(t)-\Bar{\alpha}^*(t)\right),
\end{equation}
we can check that $\left(\breve{\alpha}^*(t),\Bar{\sigma}(t),\breve{u}^*(t),\breve{y}^*(t)\right)$ is a feasible solution of Problem \eqref{robust data-driven mpc scheme}. The corresponding objective function value is denoted by $\Bar{J}_L(\Tilde{\xi}_t)$. Similarly, we have the following relation:
\begin{equation}\label{relationship 2 Proof of Proposition 1}
    \Bar{J}^*_L(\Tilde{\xi}_t) \leq \Bar{J}_L(\Tilde{\xi}_t) = \breve{J}^*_L(\xi_t) + \frac{\lambda_\sigma}{\Bar{\varepsilon}^{\beta_\sigma}}\|\Bar{\sigma}(t)\|^2.
\end{equation}
Now we focus on Problem \eqref{robust data-driven mpc scheme}. It is worth noting that we can rewrite Problem \eqref{robust data-driven mpc scheme} in an equivalent way. For simplicity, we use $H_{u\Tilde{y}}$ to represent $\left[\begin{smallmatrix}
            H_{L+l}(u^d)\\
            H_{L+l}(\Tilde{y}^d)
        \end{smallmatrix}\right]$. We can then find a matrix $Z$ whose column space spans the left-null space of $\left[\begin{smallmatrix}
            H_{L+l}(u^d)\\
            H_{L+l}(\Tilde{y}^d)
        \end{smallmatrix}\right]$, i.e.,
\begin{equation}
    \mathrm{col}(Z) = \mathrm{ker}\left(H_{u\Tilde{y}}^\top\right).
\end{equation}
By using $Z$, we can replace the constraint \eqref{parametrization using offline data} with
\begin{equation}\label{replace parametrization using offline data}
    Z^\top\begin{bmatrix}
            \Bar{u}(t)\\
            \Bar{y}(t) + \Bar{\sigma}(t)
        \end{bmatrix}=0.
\end{equation}
This constraint ensures that $\left[\begin{smallmatrix}
            \Bar{u}(t)\\
            \Bar{y}(t) + \Bar{\sigma}(t)
        \end{smallmatrix}\right]$ lies in the column space of $H_{u\Tilde{y}}$. Consequently, we can replace $\bar{\alpha}$ in the objective function \eqref{objective function} by the following closed-form expression:
\begin{equation}
    \Bar{\alpha}(t) = H_{u\Tilde{y}}^\dagger\begin{bmatrix}
            \Bar{u}(t)\\
            \Bar{y}(t) + \Bar{\sigma}(t)
        \end{bmatrix}.
\end{equation}
This is because the optimal solution $\Bar{\alpha}^*(t)$ satisfying \eqref{parametrization using offline data} has minimal norm. Now we have equivalently rewritten Problem \eqref{robust data-driven mpc scheme} by removing the variable $\Bar{\alpha}(t)$. We refer to the resulting problem as the reduced problem. To facilitate discussion, we denote the objective function value of the reduced problem corresponding to $(\bar{\sigma},\breve{u}^*(t),\breve{y}^*(t))$ by $\bar{J}_{L,red}(\tilde{\xi}_t)$. Since $R\succ0, Q\succ0$, and $\lambda_\sigma>0$, the Hessian matrix of the objective function \eqref{objective function} with respect to $\left[\begin{smallmatrix}\Bar{\sigma}(t)^\top &\Bar{u}_{[0,L-1]}(t)^\top &\Bar{y}_{[0,L-1]}(t)^\top \end{smallmatrix}\right]^\top$ satisfies
\begin{equation}
    \nabla^2 \Bar{J} \succeq \mu I,
\end{equation}
for some constant $\mu>0$. Since the initial components of $\bar{u}^*(t)$ and $\breve{u}^*(t)$ are equal, by using the strong convexity property, there exists a constant $c_{u,1}= \frac{2}{\mu}>0$ such that 
\begin{equation}
\begin{aligned}
    \|\breve{u}^*(t)-\Bar{u}^*(t)\|^2 &\leq c_{u,1}\left(\Bar{J}_{L,red}(\Tilde{\xi}_t)-\Bar{J}^*_L(\Tilde{\xi}_t)\right),\\
    &\leq c_{u,1}\left(\Bar{J}_L(\Tilde{\xi}_t)-\Bar{J}^*_L(\Tilde{\xi}_t)\right)
\end{aligned}
\end{equation}
hold. Then, 
\begin{equation}\label{relationship 3 Proof of Proposition 1}
    \begin{aligned}
        &\|\breve{u}^*(t)-\Bar{u}^*(t)\|^2 \\
        &\leq c_{u,1}\left(\breve{J}^*_L(\xi_t)+\frac{\lambda_\sigma}{\Bar{\varepsilon}^{\beta_\sigma}}\|\Bar{\sigma}(t)\|^2-\breve{J}_L(\xi_t)-\frac{\lambda_\sigma}{\Bar{\varepsilon}^{\beta_\sigma}}\|\Bar{\sigma}^*(t)\|^2\right)\\
        &\leq c_{u,1}\left(\frac{\lambda_\sigma}{\Bar{\varepsilon}^{\beta_\sigma}}\|\Bar{\sigma}(t)\|^2-\frac{\lambda_\sigma}{\Bar{\varepsilon}^{\beta_\sigma}}\|\Bar{\sigma}^*(t)\|^2\right)\\
        &\leq c_{u,1}\frac{\lambda_\sigma}{\Bar{\varepsilon}^{\beta_\sigma}}\left(\|\Bar{\sigma}(t)-\Bar{\sigma}^*(t)\|^2+2\|\Bar{\sigma}(t)-\Bar{\sigma}^*(t)\|\|\Bar{\sigma}^*(t)\|\right).
    \end{aligned}
\end{equation}
Note that 
\begin{equation}\label{relationship 4 Proof of Proposition 1}
    \Bar{\sigma}(t) - \Bar{\sigma}^*(t) = H_{L+l}(\varepsilon^d)\left(\breve{\alpha}^*(t)-\Bar{\alpha}^*(t)\right)
\end{equation}
and 
\begin{equation}\label{relationship 5 Proof of Proposition 1}
    \|\breve{\alpha}^*(t)-\Bar{\alpha}^*(t)\|^2 \leq 2\|\breve{\alpha}^*(t)\|^2 + 2\|\Bar{\alpha}^*(t)\|^2.
\end{equation}
Then, the relevant terms can be bounded as
\begin{equation}\label{relationship 6 Proof of Proposition 1}
    \|\Bar{\sigma}^*(t)\|^2 \leq \frac{\Bar{J}^*_L(\Tilde{\xi}_t)}{\lambda_\sigma}\Bar{\varepsilon}^{\beta_\sigma}, \|\Bar{\alpha}^*(t)\|^2 \leq \frac{\Bar{J}^*_L(\Tilde{\xi}_t)}{\lambda_\alpha \Bar{\varepsilon}^{\beta_\alpha}},
\end{equation}
and
\begin{equation}
    \|\breve{\alpha}^*(t)\|^2 \leq \frac{\breve{J}^*_L(\xi_t)}{\lambda_\alpha \Bar{\varepsilon}^{\beta_\alpha}} \leq \frac{\Bar{J}^*_L(\Tilde{\xi}_t)-\frac{\lambda_\sigma}{\Bar{\varepsilon}^{\beta_\sigma}}\|\Bar{\sigma}^*(t)\|^2}{\lambda_\alpha \Bar{\varepsilon}^{\beta_\alpha}}\leq \frac{\Bar{J}^*_L(\Tilde{\xi}_t)}{\lambda_\alpha \Bar{\varepsilon}^{\beta_\alpha}}.
\end{equation}
By combining \eqref{relationship 3 Proof of Proposition 1}, \eqref{relationship 4 Proof of Proposition 1}, \eqref{relationship 5 Proof of Proposition 1}, and \eqref{relationship 6 Proof of Proposition 1}, we obtain
\begin{equation}
\begin{aligned}
    &\|\breve{u}^*(t)-\Bar{u}^*(t)\|^2 \leq \\
    & c_{u,1}\Bar{J}^*_L(\Tilde{\xi}_t)\left(\frac{4c_H\lambda_\sigma}{\lambda_\alpha}\Bar{\varepsilon}^{2-\beta_\alpha-\beta_\sigma}+4\sqrt{\frac{c_H\lambda_\sigma}{\lambda_\alpha}}\Bar{\varepsilon}^{\frac{2-\beta_\alpha-\beta_\sigma}{2}}\right).
\end{aligned}
\end{equation}
Recall that $\beta_\alpha+\beta_\sigma<2$ and that \eqref{optimal solution continuity result 1} holds. Then, there exists a $\beta_{u,1}\in \mathcal{K}_\infty$ such that
\begin{equation}\label{result 1 Proof of Proposition 1}
    \|\breve{u}^*(t)-\Bar{u}^*(t)\| \leq \beta_{u,1}(\Bar{\varepsilon}).
\end{equation}
We then bound $\|\breve{u}^*(t)-\check{u}^*(t)\|$, which involves the sensitivity analysis of the optimal input trajectory with respect to $\tilde{\sigma}$. Our analysis relies on \cite[Theorem 5.53]{bonnans2013perturbation} and we will show that our problem setting satisfies all required conditions (i)-(v).

Similar to the previous part, we can eliminate $\check\alpha(t)$ via the minimum-norm representation. The nominal MPC problem \eqref{nominal data-driven mpc scheme} is thus equivalently reduced to a convex QCQP in $(\check u(t),\check y(t))$. The reduced cost is strongly convex since \(R\succ0\) and \(Q\succ0\), which guarantees uniqueness of the optimizer (i) and implies the strong second-order sufficient condition (iv). By \cref{LICQ assumption}, LICQ holds at the nominal optimizer. Since LICQ implies MFCQ, Gollan's condition (ii) holds, and the existence of Lagrange multipliers follows (iii). Moreover, the endpoint problems corresponding to $\sigma=0$ and $\sigma = \tilde\sigma$ are feasible. By convexity, their convex combinations provide feasible solutions for all perturbations $\sigma = \theta\tilde\sigma$ with $\theta\in[0,1]$ (v). To show that the feasible set of Problem \eqref{auxiliary nominal data-driven mpc scheme} is uniformly bounded, we first bound $\tilde{\sigma}$ as 
\begin{equation}
    \begin{aligned}
        &\|\tilde{\sigma}\| \leq \|\Bar{\sigma}^*(t)-H_{L+l}(\varepsilon^d)\Bar{\alpha}^*(t)\|+\|\Tilde{y}_{[t-l,t-1]}-y_{t-l,t-1}\|\\
        & \leq \|\Bar{\sigma}^*(t)\| + \|H_{L+l}(\varepsilon^d)\|\|\Bar{\alpha}^*(t)\| + \|\varepsilon_{[t-l,t-1]}\|.
    \end{aligned}
\end{equation}
Using \eqref{relationship 6 Proof of Proposition 1}, we obtain
\begin{equation}\label{sigma tilde bound}
    \|\tilde{\sigma}\| \leq \sqrt{l}\Bar{\varepsilon} + \sqrt{\Bar{J}^*_L(\Tilde{\xi}_t)}\left(\sqrt{\frac{\Bar{\varepsilon}^{\beta_\sigma}}{\lambda_\sigma}}+ \sqrt{\frac{c_H}{\lambda_\alpha}\Bar{\varepsilon}^{2-\beta_\alpha}}\right).
\end{equation}
Recall that $\bar{\varepsilon}\leq\bar{\varepsilon}_{\max}$, $\beta_\alpha+\beta_\sigma<2$, and \eqref{optimal solution continuity result 1} holds. Then, if $\bar{\varepsilon}_{\max}$ is sufficiently small, $\|\tilde{\sigma}\|$ is small enough and uniformly bounded. Together with the facts that $\left[\begin{smallmatrix}
    \check u(t)\\
    \check y(t)+\theta\tilde\sigma_1
\end{smallmatrix}\right]$ is a finite-horizon nominal trajectory, the initial input-output trajectory $\left[\begin{smallmatrix}
    \check u_{[-l,-1]}(t)\\
    \check y_{[-l,-1]}(t)
\end{smallmatrix}\right]$ is bounded, the input constraint set is compact, and $\check{\xi}_L(t)$ satisfies the terminal constraint, we can conclude that the feasible sets of Problem \eqref{auxiliary nominal data-driven mpc scheme} are uniformly bounded (v). Therefore, by \cite[Theorem 5.53]{bonnans2013perturbation}, the optimal solution of the reduced auxiliary problem is locally Lipschitz continuous with respect to \(\tilde{\sigma}\). Since the above regularity conditions hold throughout the compact parameter set defined by sufficiently small \(\|\xi_t\|\) and \(\bar\varepsilon\in[0,\bar\varepsilon_{\max}]\), there exists a uniform constant $c_{u,2}>0$ such that
\begin{equation}
\|\breve{u}^*(t)-\check{u}^*(t)\| \leq c_{u,2}\|\tilde{\sigma}\|.
\end{equation}
Combining this result with \eqref{sigma tilde bound}, there exists a $\beta_{u,2}\in \mathcal{K}_\infty$ such that
\begin{equation}\label{result 2 Proof of Proposition 1}
    \|\breve{u}^*(t)-\check{u}^*(t)\| \leq \beta_{u,2}(\Bar{\varepsilon}).
\end{equation}
Combining \eqref{result 1 Proof of Proposition 1} and \eqref{result 2 Proof of Proposition 1}, there exists a $\beta_u\in\mathcal{K}_\infty$ such that 
\begin{equation}
        \|\Bar{u}^*(t)-\check{u}^*(t)\| \leq \beta_u(\Bar{\varepsilon}),
\end{equation}
which completes the proof of the first result of \eqref{optimal solution continuity result 2}. The bound $\|\Bar{y}^*(t)-\check{y}^*(t)\|\leq \beta_y(\Bar{\varepsilon})$ follows analogously.  $\hfill \square$
\subsection{Proof of \cref{recursive feasibility proposition}}\label{proof of recursive feasibility proposition}
In this section, we prove the recursive feasibility of Problem \eqref{robust data-driven mpc scheme} by constructing a feasible candidate solution at time $t+1$ based on the optimal solution at time $t$. Given that \cref{terminal conditions assumption} holds and the SDP \eqref{robust terminal set design sdp} is feasible, we can compute the weight matrix $P$, the terminal controller $K$, and the terminal set $\Sigma_f$ that satisfy the conditions in \cref{terminal conditions assumption}. Then, we can construct 
\begin{equation}\label{uy bar Proof of Proposition 2}
    \begin{aligned}
    &\Bar{u}_{[-l,-1]}(t+1) = u_{[t+1-l,t]},\\
    & \Bar{u}_{[0,L-2]}(t+1) = \Bar{u}^*_{[1,L-1]}(t),\\
    & \Bar{u}_{L-1}(t+1) = K\Bar{\xi}^*_L(t),\\
    & \Bar{y}_{[-l,-1]}(t+1) = \Tilde{y}_{[t+1-l,t]},\\
    & \Bar{y}_{[0,L-2]}(t+1) = \Bar{y}^*_{[1,L-1]}(t),\\
    & \Bar{y}_{L-1}(t+1) = (\Tilde{C}+\Tilde{D}K)\Bar{\xi}^*_L(t).\\
    \end{aligned}
\end{equation}
According to the first and second conditions in \cref{terminal conditions assumption}, our candidate solution satisfies $\Bar{u}_k(t+1)\in\mathbb{U}, k\in\mathbb{I}_{[0,L-1]}$, and $\Bar{\xi}_L(t+1)\in\Sigma_f$. Then, we construct
\begin{equation}\label{alpha bar Proof of Proposition 2}
    \Bar{\alpha}(t+1) = H_{ux}^\dagger
    \begin{bmatrix}
        \Bar{u}(t+1)\\
        x_{t+1-l}
    \end{bmatrix},
\end{equation}
and 
\begin{equation}
    \Bar{\sigma}(t+1) = H_{L+l}\left(\Tilde{y}^d\right)\Bar{\alpha}(t+1)-\Bar{y}(t+1).
\end{equation}
Our candidate solution $\Bar{\alpha}(t+1),\Bar{\sigma}(t+1),\Bar{u}(t+1),\Bar{y}(t+1)$ satisfies all constraints of Problem \eqref{robust data-driven mpc scheme} and is thus feasible. The proof is complete. $\hfill \square$
\subsection{Proof of \cref{practical exponential stability theorem}}\label{proof of practical exponential stability theorem}
In this section, we prove the practical exponential stability of the closed-loop system under the proposed robust data-driven MPC scheme. The proof integrates all the technical results in the previous sections. First, by leveraging the candidate solution at time $t+1$, we can bound the Lyapunov function decrease as follows:
\begin{equation}\label{lyapunov decrease 1 Proof of Theorem 3}
\begin{aligned}
    &V_{t+1} - V_t \leq \sum_{k=0}^{L-1}\left(\|\Bar{u}_k(t+1)\|^2_R + \|\Bar{y}_k(t+1)\|^2_Q \right)+ \\
    &\|\Bar{\xi}_L(t+1)\|^2_P + \lambda_\alpha\Bar{\varepsilon}^{\beta_\alpha}\|\Bar{\alpha}(t+1)\|^2 + \frac{\lambda_\sigma}{\Bar{\varepsilon}^{\beta_\sigma}}\|\Bar{\sigma}(t+1)\|^2 - \\
    &\sum_{k=0}^{L-1}\left(\|\Bar{u}^*_k(t)\|^2_R + \|\Bar{y}^*_k(t)\|^2_Q\right) - \|\Bar{\xi}^*_L(t)\|^2_P - \lambda_\alpha\Bar{\varepsilon}^{\beta_\alpha}\|\Bar{\alpha}^*(t)\|^2 \\
    &- \frac{\lambda_\sigma}{\Bar{\varepsilon}^{\beta_\sigma}}\|\Bar{\sigma}^*(t)\|^2 + \gamma \left(W(\xi_{t+1})-W(\xi_t)\right).
\end{aligned}
\end{equation}
Since the candidate solution $\Bar{u}(t+1)$ and $\Bar{y}(t+1)$ are directly shifted from $\Bar{u}^*(t)$ and $\Bar{y}^*(t)$, we can simplify \eqref{lyapunov decrease 1 Proof of Theorem 3} as
\begin{equation}\label{lyapunov decrease 2 Proof of Theorem 3}
\begin{aligned}
    &V_{t+1} - V_t \leq \|\Bar{u}_{L-1}(t+1)\|^2_R + \|\Bar{y}_{L-1}(t+1)\|^2_Q + \\
    &\|\Bar{\xi}_L(t+1)\|^2_P + \lambda_\alpha\Bar{\varepsilon}^{\beta_\alpha}\|\Bar{\alpha}(t+1)\|^2 + \frac{\lambda_\sigma}{\Bar{\varepsilon}^{\beta_\sigma}}\|\Bar{\sigma}(t+1)\|^2 - \\
    &\|\Bar{u}^*_0(t)\|^2_R - \|\Bar{y}^*_0(t)\|^2_Q - \|\Bar{\xi}^*_L(t)\|^2_P  +\gamma \left(W(\xi_{t+1})-W(\xi_t)\right),
\end{aligned}
\end{equation}
which can be further simplified by using the third condition of \cref{terminal conditions assumption}
\begin{equation}\label{lyapunov decrease 3 Proof of Theorem 3}
\begin{aligned}
    &V_{t+1} - V_t\leq-\|\Bar{u}^*_0(t)\|^2_R - \|\Bar{y}^*_0(t)\|^2_Q +\frac{\lambda_\sigma}{\Bar{\varepsilon}^{\beta_\sigma}}\|\Bar{\sigma}(t+1)\|^2   \\
    & +\lambda_\alpha\Bar{\varepsilon}^{\beta_\alpha}\|\Bar{\alpha}(t+1)\|^2 + \gamma \left(W(\xi_{t+1})-W(\xi_t)\right).
\end{aligned}
\end{equation}
The rest of the proof focuses on bounding the right-hand side of \eqref{lyapunov decrease 3 Proof of Theorem 3}. We first bound $\|\Bar{\alpha}(t+1)\|^2$. Recall that $\Bar{\alpha}(t+1)$ is constructed in \eqref{alpha bar Proof of Proposition 2}. Then, 
\begin{equation}\label{relationship 1 Proof of Theorem 3}
\|\Bar{\alpha}(t+1)\|^2 \leq \|H_{ux}^\dagger\|^2\left(\|\Bar{u}(t+1)\|^2+\|x_{t+1-l}\|^2\right).  
\end{equation}
The candidate solution $\Bar{u}(t+1)$ in \eqref{uy bar Proof of Proposition 2} enables us to decompose $\|\Bar{u}(t+1)\|^2$ as
\begin{equation}\label{relationship 2 Proof of Theorem 3}
\|\Bar{u}(t+1)\|^2 = \|\Bar{u}^*_{[-l+1,L-1]}(t)\|^2+\|\Bar{u}_{L-1}(t+1)\|^2.
\end{equation}
Note that we can explicitly express $x_{t-l+1}$ by using $\xi_t$
\begin{equation}
    x_{t-l+1} = 
    \underbrace{\begin{bmatrix}
        -A\mathcal{O}_l^\dagger\mathcal{T}_l + \left[\begin{smallmatrix}B &0 &\cdots &0\end{smallmatrix}\right] &A\mathcal{O}^\dagger_l
    \end{bmatrix}}_{M}\xi_t.
\end{equation}
Therefore, 
\begin{equation}\label{relationship 3 Proof of Theorem 3}
    \|x_{t+1-l}\|^2 \leq \|M\|^2\|\xi_t\|^2.
\end{equation}
Besides, since $\Bar{\xi}^*_L(t)\in\Sigma_f$, the candidate solution in \eqref{uy bar Proof of Proposition 2} implies that $\Bar{\xi}_L(t+1)\in\Sigma_f$, according to the terminal conditions in \cref{terminal conditions assumption}. Then, we have
\begin{equation}\label{relationship 4 Proof of Theorem 3}
    \|\Bar{u}_{L-1}(t+1)\|^2  \leq \|\Bar{\xi}_L(t+1)\|^2 \leq \frac{r}{\lambda_{\min}(P)}.
\end{equation}
Furthermore, we can use $\Bar{V}$, i.e., the upper bound of $V_t$, to bound $ \|\Bar{u}^*_{[-l+1,L-1]}(t)\|^2$ as
\begin{equation}\label{relationship 5 Proof of Theorem 3}
\begin{aligned}
     \|\Bar{u}^*_{[-l+1,L-1]}(t)\|^2 &= \|\Bar{u}^*_{[-l+1,-1]}(t)\|^2 + \|\Bar{u}^*_{[0,L-1]}(t)\|^2\\
     &\leq \|\xi_t\|^2 + \frac{\bar{V}}{\lambda_{\min}(R)}.
\end{aligned}
\end{equation}
By substituting \eqref{relationship 2 Proof of Theorem 3}, \eqref{relationship 3 Proof of Theorem 3}, \eqref{relationship 4 Proof of Theorem 3}, and \eqref{relationship 5 Proof of Theorem 3} into \eqref{relationship 1 Proof of Theorem 3}, we obtain
\begin{equation}\label{relationship 6 Proof of Theorem 3}
\begin{aligned}
    &\|\Bar{\alpha}(t+1)\|^2 \leq \\&\|H_{ux}^\dagger\|^2\left(\frac{\Bar{V}}{\lambda_{\min}(R)}+\frac{r}{\lambda_{\min}(P)}+\left(\|M\|^2+1\right)\|\xi_t\|^2\right).
\end{aligned}
\end{equation}
We now consider bounding $\|\Bar{\sigma}(t+1)\|^2$. By using our constructed $\Bar{\sigma}(t+1)$, we have
\begin{equation}\label{relationship 7 Proof of Theorem 3}
    \begin{aligned}
        &\|\Bar{\sigma}(t+1)\|^2\\
        &= \|H_{L+l}(\varepsilon^d)\bar{\alpha}(t+1)+H_{L+l}(y^d)\Bar{\alpha}(t+1)-\Bar{y}(t+1)\|^2\\
        & \leq 2c_H \Bar{\varepsilon}^2\|\bar{\alpha}(t+1)\|^2 + 2\|H_{L+l}(y^d)\Bar{\alpha}(t+1)-\Bar{y}(t+1)\|^2\\
        & \leq 2c_H \Bar{\varepsilon}^2\|\bar{\alpha}(t+1)\|^2 + 2\|\Hat{y}_{[t-l+1,t+L]}-\Bar{y}(t+1)\|^2.
    \end{aligned}
\end{equation}
The second term can be decomposed as 
\begin{equation}\label{relationship 8 Proof of Theorem 3}
\begin{aligned}
    &\|\Hat{y}_{[t-l+1,t+L]}-\Bar{y}(t+1)\|^2 \\
    &\leq 2l\Bar{\varepsilon}^2+2\|\Hat{y}_{[t+1,t+L]}-\Bar{y}_{[0,L-1]}(t+1)\|^2\\
    & =  2l\Bar{\varepsilon}^2+2\sum^{L-1}_{i=1}\|\hat{y}_{t+i}-\bar{y}^*_i(t)\|^2+2\|\hat{y}_{t+L}-\bar{y}_{L-1}(t+1)\|^2.
\end{aligned}
\end{equation}
Note that in \eqref{relationship 7 Proof of Theorem 3}, the bound on $\|\Bar{\alpha}(t+1)\|^2$ has been derived in \eqref{relationship 6 Proof of Theorem 3}. Then, according to \eqref{relationship 8 Proof of Theorem 3}, we only need to bound $\sum^{L-1}_{i=1}\|\hat{y}_{t+i}-\bar{y}^*_i(t)\|^2$ and $\|\hat{y}_{t+L}-\bar{y}_{L-1}(t+1)\|^2$ to obtain a complete bound on $\|\Bar{\sigma}(t+1)\|^2$. 

Recall from \cref{prediction error lemma} that, for $i\in\mathbb{I}_{[1,L-1]}$, the output prediction error satisfies
    \begin{equation}\label{output prediction bound Proof of Theorem 3}
        \begin{aligned}
        &\|\hat{y}_{t+i}-\Bar{y}^*_i(t)\|^2 \leq 3c_N\Bar{\varepsilon}^2\|\Bar{\alpha}^*(t)\|^2 + 3\|\Bar{\sigma}^*_i(t)\|^2\\
        &+3\rho_{2,l+i}(3l\Bar{\varepsilon}^2+3\|\Bar{\sigma}^*_{[-l,-1]}(t)\|^2+
        3lc_N\Bar{\varepsilon}^2\|\Bar{\alpha}^*(t)\|^2).
        \end{aligned}
    \end{equation}
In \eqref{output prediction bound Proof of Theorem 3}, only $\|\Bar{\sigma}_i^*(t)\|^2$ and $\|\Bar{\sigma}^*_{[-l,-1]}(t)\|^2$ do not explicitly vanish as $\Bar{\varepsilon}\rightarrow 0$. In what follows, we will provide a bound on $\|\Bar{\sigma}^*(t)\|^2$ to show that $\|\Bar{\sigma}^*(t)\|^2$ not only vanishes as $\Bar{\varepsilon}\rightarrow 0$ but can also be effectively incorporated into \eqref{lyapunov decrease 3 Proof of Theorem 3} to ensure the Lyapunov function decrease.

First, according to \cref{lower and upper bound lemma}, the choice $\Bar{V} = \gamma\lambda_{\min}(P_W)\delta^2$ implies that $\|\xi_t\|^2\leq \frac{\Bar{V}}{\gamma\lambda_{\min}(P_W)} = \delta^2$. This fact, together with $\|x_t\|\leq \|\Gamma\|\|\xi_t\|$, further implies that if $\delta$ is sufficiently small, $x_t$ is sufficiently close to the origin. Then, by controllability, there is an input trajectory that steers the system state to the terminal set while satisfying input constraints. Consequently, Problem \eqref{nominal data-driven mpc scheme} is feasible. We then utilize the optimal solution of Problem \eqref{nominal data-driven mpc scheme} to construct a feasible solution of Problem \eqref{robust data-driven mpc scheme}, i.e.,
\begin{equation}
    \Bar{\alpha}(t) = \Check{\alpha}^*(t), \Bar{u}(t) = \Check{u}^*(t), \Bar{y}(t) = \begin{bmatrix}
        \tilde{y}_{[t-l,t-1]}\\
        \Check{y}_{[0,L-1]}^*(t)
    \end{bmatrix}.
\end{equation}
To satisfy the equality constraint \eqref{parametrization using offline data}, we choose $\Bar{\sigma}(t)$ as
\begin{equation}
    \begin{aligned}
        \Bar{\sigma}_{[-l,-1]}(t) &= H_l\left(\varepsilon^d_{[0,N-L-1]}\right)\bar{\alpha}(t)-\varepsilon_{[t-l,t-1]},\\
        \Bar{\sigma}_{[0,L-1]}(t) &= H_L\left(\varepsilon^d_{[l,N-1]}\right)\bar{\alpha}(t).
    \end{aligned}
\end{equation}
As in the proof of \cref{lower and upper bound lemma}, we have
\begin{equation}\label{relationship 9 Proof of Theorem 3}
\|\Bar{\sigma}(t)\|^2\leq 2c_H\Bar{\varepsilon}^2\|\Bar{\alpha}(t)\|^2 + 2l\Bar{\varepsilon}^2.
\end{equation}
To facilitate subsequent analysis, we define
\begin{equation}
\begin{aligned}
    &H(\Bar{u}(t),\Bar{y}(t),\Bar{\alpha}(t)) \triangleq \\
    &\sum_{k=0}^{L-1}\left(\|\Bar{u}_k(t)\|_R^2 + \|\Bar{y}_k(t)\|_Q^2\right)+\|\Bar{\xi}_L(t)\|^2_P + \lambda_\alpha \Bar{\varepsilon}^{\beta_\alpha}\|\Bar{\alpha}(t)\|^2.
\end{aligned}
\end{equation}
Then, according to optimality:
\begin{equation}
\begin{aligned}
    &H(\Bar{u}^*(t),\Bar{y}^*(t),\Bar{\alpha}^*(t))+\frac{\lambda_\sigma}{\bar{\varepsilon}^{\beta_\sigma}}\|\bar{\sigma}^*(t)\|^2 \leq\\
    &H(\Bar{u}(t),\Bar{y}(t),\Bar{\alpha}(t))+\frac{\lambda_\sigma}{\bar{\varepsilon}^{\beta_\sigma}}\|\bar{\sigma}(t)\|^2,
\end{aligned}
\end{equation}
which implies 
\begin{equation}\label{sigma bar * bound Proof of Theorem 3}
\begin{aligned}
    &\|\bar{\sigma}^*(t)\|^2 \leq \|\bar{\sigma}(t)\|^2 +\\
    &\left(H(\Bar{u}(t),\Bar{y}(t),\Bar{\alpha}(t)) - H(\Bar{u}^*(t),\Bar{y}^*(t),\Bar{\alpha}^*(t)) \right)\frac{\bar{\varepsilon}^{\beta_\sigma}}{\lambda_\sigma}.
\end{aligned}
\end{equation}
By some calculations, we can obtain
\begin{equation}
\begin{aligned}
    &H(\Bar{u}(t),\Bar{y}(t),\Bar{\alpha}(t)) - H(\Bar{u}^*(t),\Bar{y}^*(t),\Bar{\alpha}^*(t))=\\
    &  \|\bar{u}_{[0,L-1]}(t)\|^2_{I_L\otimes R}-\|\bar{u}_{[0,L-1]}^*(t)\|^2_{I_L\otimes R} +
    \|\bar{y}_{[0,L-1]}(t)\|^2_{I_L\otimes Q}\\
    &-\|\bar{y}_{[0,L-1]}^*(t)\|^2_{I_L\otimes Q} + \|\bar{\xi}_L(t)\|^2_{P}-\|\bar{\xi}_L^*(t)\|^2_{P} \\
    &+ \lambda_\alpha \bar{\varepsilon}^{\beta_\alpha}\|\bar{\alpha}(t)\|^2_2-\lambda_\alpha \bar{\varepsilon}^{\beta_\alpha}\|\bar{\alpha}^*(t)\|^2_2\\
    &
    \leq \lambda_{\max}(R)\left(\|\bar{u}(t)-\bar{u}^*(t)\|^2 + 2\|\bar{u}(t)-\bar{u}^*(t)\|\|\bar{u}^*(t)\|\right)\\
    &+\lambda_{\max}(Q)\left(\|\bar{y}(t)-\bar{y}^*(t)\|^2 + 2\|\bar{y}(t)-\bar{y}^*(t)\|\|\bar{y}^*(t)\|\right)\\
    &+\|\bar{\xi}_L(t)-\bar{\xi}^*_L(t)\|^2_{P} + 2\|\bar{\xi}_L(t)-\bar{\xi}^*_L(t)\|_P\|\bar{\xi}^*_L(t)\|_{P}\\
    &+\lambda_\alpha \bar{\varepsilon}^{\beta_\alpha}\|\bar{\alpha}(t)\|^2_2.
\end{aligned}
\end{equation}
Since $\Bar{V} = \gamma\lambda_{\min}(P_W)\delta^2$ implies $\|\xi_t\|\leq \delta$, we can use the result \eqref{optimal solution continuity result 2} to bound $\|\bar{u}(t)-\bar{u}^*(t)\|$ and $\|\bar{y}(t)-\bar{y}^*(t)\|$. The term $\|\bar{\xi}_L(t)-\bar{\xi}_L^*(t)\|_P$ can be bounded in a similar way due to the structure of $\bar{\xi}_L(t)$ and $\bar{\xi}^*_L(t)$. In addition, $\|\bar{u}^*(t)\|, \|\bar{y}^*(t)\|$, and $\|\bar{\xi}_L^*(t)\|_{P}$ can be bounded using $\Bar{V}$. In summary, there exist functions $\beta_{H,1}\in \mathcal{K}_\infty$ and $\beta_{H,2}\in \mathcal{K}_\infty$ such that 
\begin{equation}\label{H difference bound Proof of Theorem 3}
    \begin{aligned}
        &H(\Bar{u}(t),\Bar{y}(t),\Bar{\alpha}(t)) - H(\Bar{u}^*(t),\Bar{y}^*(t),\Bar{\alpha}^*(t))\\
        &\leq \beta_{H,1}(\bar{\varepsilon})+\beta_{H,2}(\bar{\varepsilon})\sqrt{\bar{V}} + \lambda_\alpha \bar{\varepsilon}^{\beta_\alpha}\|\bar{\alpha}(t)\|^2.
    \end{aligned}
\end{equation}
We next focus on bounding $\lambda_\alpha \bar{\varepsilon}^{\beta_\alpha}\|\bar{\alpha}(t)\|^2_2$. Recall from \eqref{alpha check expression} that $\bar{\alpha}(t)=\check{\alpha}^*(t)$ is given by
\begin{equation}
    \bar{\alpha}(t) = H_{ux}^\dagger \begin{bmatrix}\check{u}^*(t)\\x_{t-l}\end{bmatrix}.
\end{equation}
Then, we obtain 
\begin{equation}\label{alpha bar bound Proof of Theorem 3}
    \begin{aligned}
        &\|\Bar{\alpha}(t)\|^2 \leq \|H_{ux}^\dagger\|^2\left(\|\check{u}^*_{[-l,L-1]}(t)\|^2+\|x_{t-l}\|^2\right)\\
        &= \|H_{ux}^\dagger\|^2\left(\|\check{u}^*_{[0,L-1]}(t)\|^2+\|\check{u}^*_{[-l,-1]}(t)\|^2 + \|x_{t-l}\|^2\right)\\
        &\leq \|H_{ux}^\dagger\|^2\left(\frac{\check{J}^*_L(\xi_t)}{\lambda_{\min}(R)} + \left\lVert\left[\begin{smallmatrix}I &0\\-\mathcal{O}_l^\dagger\mathcal{T}_l &\mathcal{O}_l^\dagger\end{smallmatrix}\right]\right\rVert^2\|\xi_t\|^2\right).
    \end{aligned}
\end{equation}
Since $\check{J}^*_L(\xi_t)$ is the optimal cost of the nominal MPC problem, it can be bounded using the candidate solution constructed in the proof of \cref{lower and upper bound lemma}. Here, we omit the details of the procedure and directly provide the bound on $\check{J}^*_L(\xi_t)$:
\begin{equation}\label{J check bound Proof of Theorem 3}
    \check{J}^*_L(\xi_t) \leq \left(c_{s,\xi} + \lambda_\alpha \bar{\varepsilon}^{\beta_\alpha}c_{\alpha,\xi}\right)\|\xi_t\|^2.
\end{equation}
From \eqref{alpha bar bound Proof of Theorem 3}, \eqref{J check bound Proof of Theorem 3}, and $\Bar{\varepsilon}\leq
\Bar{\varepsilon}_{\max}$, there exists a constant $c_{\bar{\alpha}}>0$ such that 
\begin{equation}\label{alpha bar final bound Proof of Theorem 3}
    \|\Bar{\alpha}(t)\|^2 \leq c_{\bar{\alpha}} \|\xi_t\|^2
\end{equation}
holds. Combining \eqref{relationship 9 Proof of Theorem 3}, \eqref{sigma bar * bound Proof of Theorem 3}, \eqref{H difference bound Proof of Theorem 3}, and \eqref{alpha bar final bound Proof of Theorem 3}, we obtain 
\begin{equation}\label{relationship 10 Proof of Theorem 3}
\begin{aligned}
    &\|\bar{\sigma}^*(t)\|^2 \leq 2c_H\Bar{\varepsilon}^2 c_{\bar{\alpha}} \|\xi_t\|^2 + 2l\Bar{\varepsilon}^2+\\
    &\left(\beta_{H,1}(\bar{\varepsilon})+\beta_{H,2}(\bar{\varepsilon})\sqrt{\bar{V}}+\lambda_\alpha \bar{\varepsilon}^{\beta_\alpha}c_{\bar{\alpha}} \|\xi_t\|^2\right)\frac{\bar{\varepsilon}^{\beta_\sigma}}{\lambda_\sigma}.
\end{aligned}
\end{equation}
We combine \eqref{output prediction bound Proof of Theorem 3} and \eqref{relationship 10 Proof of Theorem 3}, and use the fact that $\|\bar{\sigma}_i^*(t)\|^2\leq\|\bar{\sigma}^*(t)\|^2$, $\|\bar{\sigma}_{[-l,-1]}^*(t)\|^2\leq\|\bar{\sigma}^*(t)\|^2$, and $\|\bar{\alpha}^*(t)\|^2\leq\frac{\bar{V}}{\lambda_\alpha \bar{\varepsilon}^{\beta_\alpha}}$. Then, we can bound $\sum^{L-1}_{i=1}\|\hat{y}_{t+i}-\bar{y}^*_i(t)\|^2$ as
\begin{equation}\label{sum term  bound Proof of Theorem 3}
\begin{aligned}
    &\sum^{L-1}_{i=1}\|\hat{y}_{t+i}-\bar{y}^*_i(t)\|^2 \leq c_1\bar{\varepsilon}^{2-\beta_\alpha}+c_2 \bar{\varepsilon}^2+c_3\bar{\varepsilon}^{2-\beta_\alpha}+\big(c_4\Bar{\varepsilon}^2+\\
    &\left(c_5\beta_{H,1}(\bar{\varepsilon})+c_6\beta_{H,2}(\bar{\varepsilon})+c_7\bar{\varepsilon}^{\beta_\alpha} \|\xi_t\|^2\right)\bar{\varepsilon}^{\beta_\sigma} + c_8\Bar{\varepsilon}^2 \|\xi_t\|^2 \big),
\end{aligned}
\end{equation}
for some constants $c_i, i \in \mathbb{I}_{[1,8]}$.
Our next step is to bound $\|\hat{y}_{t+L}-\bar{y}_{L-1}(t+1)\|^2$. We first compute the nominal state at time $t+L-l$ as
\begin{equation}
    \hat{x}_{t+L-l} = \begin{bmatrix}-\mathcal{O}^\dagger_l \mathcal{T}_l &\mathcal{O}^\dagger_l\end{bmatrix} \Hat{\xi}_{t+L}, \Hat{\xi}_{t+L} \triangleq \begin{bmatrix}\Bar{u}^*_{[L-l,L-1]}(t) \\ \Hat{y}_{[t+L-l,t+L-1]}\end{bmatrix}.
\end{equation}
Then, $\Hat{y}_{t+L}$ is the output starting from $\hat{x}_{t+L-l}$ with the application of $\Bar{u}^*_{[L-l,L-1]}(t)$ and $\Bar{u}_{L-1}(t+1)$. To describe $\bar{y}_{L-1}(t+1)$, we note that $\bar{\xi}^*_L(t)$ is not a trajectory of \eqref{basic system}. However, we can calculate a state as 
\begin{equation}\label{virtual state x barbar}
\bar{\bar{x}} = \begin{bmatrix}-\mathcal{O}^\dagger_l \mathcal{T}_l &\mathcal{O}^\dagger_l\end{bmatrix} \Bar{\xi}^*_L(t).
\end{equation}
We then let $\bar{\bar{y}}_{t+L-l,t+L-1}$ denote the output starting from $\bar{\bar{x}}$ with the application of $\Bar{u}^*_{[L-l,L-1]}(t)$. Then, since in \eqref{virtual state x barbar}, $\bar{\bar{x}}$ is computed in the least-squares sense, it represents the state that minimizes $\|\bar{\bar{y}}_{t+L-l,t+L-1}-\Bar{y}^*_{[L-l,L-1]}(t)\|$ under the same input sequence $\bar{u}^*_{[L-l,L-1]}(t)$. 
Then, according to the construction of $\bar{y}_{L-1}(t+1)$ in \eqref{uy bar Proof of Proposition 2}, $\bar{y}_{L-1}(t+1)$ can be viewed as the output starting from $\bar{\bar{x}}$ with the application of $\Bar{u}^*_{[L-l,L-1]}(t)$ and $\Bar{u}_{L-1}(t+1)$. Consequently,
\begin{equation}\label{final term  bound Proof of Theorem 3}
    \begin{aligned}
        &\|\hat{y}_{t+L}-\bar{y}_{L-1}(t+1)\|^2 \\
        &= \|CA^l (\hat{x}_{t+L-l} - \bar{\bar{x}} )\|^2\\
        &\leq \left\lVert CA^l \begin{bmatrix}-\mathcal{O}^\dagger_l \mathcal{T}_l &\mathcal{O}^\dagger_l\end{bmatrix}\left(\hat{\xi}_{t+L}-\bar{\xi}^*_L(t)\right)\right\rVert^2\\
        & \leq \left\lVert CA^l \begin{bmatrix}-\mathcal{O}^\dagger_l \mathcal{T}_l &\mathcal{O}^\dagger_l\end{bmatrix} \right\rVert^2 \left\lVert \hat{y}_{[t+L-l,t+L-1]}-\bar{y}^*_{[L-l,L-1]}(t) \right\rVert^2\\
        & \leq \left\lVert CA^l \begin{bmatrix}-\mathcal{O}^\dagger_l \mathcal{T}_l &\mathcal{O}^\dagger_l\end{bmatrix} \right\rVert^2\sum^{L-1}_{i=1} \|\hat{y}_{t+i}-\bar{y}^*_i(t)\|^2.
    \end{aligned}
\end{equation}
Therefore, we can use the bound on $\sum^{L-1}_{i=1} \|\hat{y}_{t+i}-\bar{y}^*_i(t)\|^2$ in \eqref{sum term  bound Proof of Theorem 3} to obtain a bound for $\|\hat{y}_{t+L}-\bar{y}_{L-1}(t+1)\|^2$. By combining \eqref{relationship 6 Proof of Theorem 3}, \eqref{relationship 7 Proof of Theorem 3}, \eqref{relationship 8 Proof of Theorem 3}, \eqref{sum term  bound Proof of Theorem 3}, and \eqref{final term  bound Proof of Theorem 3}, we can bound $\frac{\lambda_\sigma}{\Bar{\varepsilon}^{\beta_\sigma}}\|\Bar{\sigma}(t+1)\|^2$ and $\lambda_\alpha\Bar{\varepsilon}^{\beta_\alpha}\|\Bar{\alpha}(t+1)\|^2$ in \eqref{lyapunov decrease 3 Proof of Theorem 3}. The remaining terms on the right-hand side of \eqref{lyapunov decrease 3 Proof of Theorem 3} can be bounded as follows:
\begin{equation}
\begin{aligned}
    &-\|\Bar{u}^*_0(t)\|^2_R - \|\Bar{y}^*_0(t)\|^2_Q + \gamma \left(W(\xi_{t+1})-W(\xi_t)\right) \\
    & \overset{\eqref{IOSS eq}}{\leq} -\|\bar{u}^*_0(t)\|^2_R - \|\bar{y}^*_0(t)\|^2_Q\\
    &+ \gamma\left(-\frac{1}{2}\|\xi_t\|^2 + c_{W,1}\|u_t\|^2 + c_{W,2}\|y_t\|^2\right).
\end{aligned}
\end{equation}
By using $
    \|y_t\|^2 \leq 2\|y_t-\bar{y}^*_0(t)\|^2 + 2\|\bar{y}^*_0\|^2,$
we can choose 
\begin{equation}
    \gamma = \frac{\lambda_{\min}(Q,R)}{\max\{c_{W,1},2c_{W,2}\}}
\end{equation}
such that
\begin{equation}\label{rest of terms Proof of Theorem 3}
\begin{aligned}
    &-\|\Bar{u}^*_0(t)\|^2_R - \|\Bar{y}^*_0(t)\|^2_Q + \gamma \left(W(\xi_{t+1})-W(\xi_t)\right)\\
    &\leq -\frac{\gamma}{2}\|\xi_t\|^2 + 2c_{W,2} \gamma \|y_t-\bar{y}^*_0(t)\|^2.
\end{aligned}
\end{equation}
Combining \eqref{relationship 6 Proof of Theorem 3}, \eqref{relationship 7 Proof of Theorem 3}, \eqref{relationship 8 Proof of Theorem 3}, \eqref{sum term  bound Proof of Theorem 3}, \eqref{final term  bound Proof of Theorem 3}, and \eqref{rest of terms Proof of Theorem 3}, the Lyapunov function decrease \eqref{lyapunov decrease 3 Proof of Theorem 3} can be bounded as 
\begin{equation}\label{simplied lyapunov decrease}
\begin{aligned}
    &V_{t+1} - V_t \leq\\ 
    &(c_9 \bar{\varepsilon}^{\beta_\alpha} + c_{10}\bar{\varepsilon}^{2-\beta_\sigma}+c_{11}\bar{\varepsilon}^{\beta_\alpha+\beta_\sigma}+c_{12}\bar{\varepsilon}^2-\frac{\gamma}{2})\|\xi_t\|^2 + \\
    &c_{13}\bar{\varepsilon}^{\beta_\alpha} + c_{14}\bar{\varepsilon}^{2-\beta_\sigma}+(c_{15}+c_{16}\bar{\varepsilon}^{\beta_\sigma})(c_{17}\bar{\varepsilon}^{2-\beta_\sigma}+\\
    &c_{18}\bar{\varepsilon}^{2-\beta_\alpha-\beta_\sigma}+c_{19}\beta_{H,1}(\bar{\varepsilon})+c_{20}\beta_{H,2}(\bar{\varepsilon})),
\end{aligned}
\end{equation}
for some constants $c_i, i\in\mathbb{I}_{[9,20]}$. According to \cref{lower and upper bound lemma}, 
\begin{equation}\label{inequality from lower and upper bound}
    \|\xi_t\|^2 \geq \frac{V_t-c_{V,2}(\bar{\varepsilon})}{c_{V,1}(\bar{\varepsilon})}.
\end{equation}
If $\bar{\varepsilon}$ is sufficiently small, the coefficient of $\|\xi_t\|^2$ will be negative. Thus, we can use \eqref{inequality from lower and upper bound} to rewrite \eqref{simplied lyapunov decrease} as 
\begin{equation}
\begin{aligned}
    &V_{t+1}\\
    &\leq \left(1+\frac{c_9 \bar{\varepsilon}^{\beta_\alpha} + c_{10}\bar{\varepsilon}^{2-\beta_\sigma}+c_{11}\bar{\varepsilon}^{\beta_\alpha+\beta_\sigma}+c_{12}\bar{\varepsilon}^2-\frac{\gamma}{2}}{c_{V,1}}\right)V_t +\\
    &\frac{c_{V,2}(\bar{\varepsilon})}{c_{V,1}(\bar{\varepsilon})}\left(\frac{\gamma}{2} - c_9 \bar{\varepsilon}^{\beta_\alpha} - c_{10}\bar{\varepsilon}^{2-\beta_\sigma}-c_{11}\bar{\varepsilon}^{\beta_\alpha+\beta_\sigma}-c_{12}\bar{\varepsilon}^2\right)+ \\
    &c_{13}\bar{\varepsilon}^{\beta_\alpha}+c_{14}\bar{\varepsilon}^{2-\beta_\sigma} + (c_{15}+c_{16}\bar{\varepsilon}^{\beta_\sigma})(c_{17}\bar{\varepsilon}^{2-\beta_\sigma}+\\
    &c_{18}\bar{\varepsilon}^{2-\beta_\alpha-\beta_\sigma}+c_{19}\beta_{H,1}(\bar{\varepsilon})+c_{20}\beta_{H,2}(\bar{\varepsilon})).
\end{aligned}
\end{equation}
Recall that as $\bar{\varepsilon}\rightarrow 0$, $c_{V,1}(\bar{\varepsilon})$ is lower-bounded by a positive constant and $c_{V,2}(\bar{\varepsilon})$ vanishes. Therefore, there exists a constant $0<\rho<1$ and a $\mathcal{K}_\infty$ function $\beta_V$ such that $V_{t+1}$ satisfies
\begin{equation}\label{lyapunov decrease 4 Proof of Theorem 3}
    V_{t+1} \leq \rho V_t + \beta_V(\bar{\varepsilon}).
\end{equation}
By defining 
\begin{equation}\label{change of variable Proof of Theorem 3}
    \beta(\bar{\varepsilon}) \triangleq \frac{1}{1-\rho}\beta_V(\bar{\varepsilon}) \rightarrow \beta_V(\bar{\varepsilon}) = (1-\rho)\beta(\bar{\varepsilon}),
\end{equation}
we obtain another $\mathcal{K}_\infty$ function $\beta$. We can thereby find a sufficiently small $\Bar{\varepsilon}_{\max}>0$ such that 
\begin{equation}
    \beta(\Bar{\varepsilon}_{\max})
\le \Bar{V}, c_{V,2}(\bar{\varepsilon}_{\max}) < \Bar{V}.
\end{equation}
The first inequality ensures that if $\Bar{\varepsilon}\leq\Bar{\varepsilon}_{\max}$, $V_t\leq\bar{V}$ implies $V_{t+1}\leq\bar{V}$. The second one, together with the upper bound in \eqref{lower and upper bound eq}, guarantees that this sublevel set contains a neighborhood of the origin. It is also worth noting that $V_{t+1} \leq \bar{V} = \gamma\lambda_{\min}(P_W)\delta^2$ implies $\|\xi_{t+1}\| \leq\delta$. This implies that \cref{lower and upper bound lemma} still holds at time $t+1$, so that the proof procedure also applies for time $t+1$. Besides, $\|\xi_{t+1}\| \leq\delta$ also indicates that the nominal MPC scheme \eqref{nominal data-driven mpc scheme} will remain feasible at time $t+1$. Thus, the relevant result, e.g., \cref{optimal solution continuity proposition}, still holds at time $t+1$. Consequently, \eqref{lyapunov decrease 4 Proof of Theorem 3}
holds recursively. To show practical exponential stability, we first derive the following from \eqref{lyapunov decrease 4 Proof of Theorem 3}
\begin{equation}
    \begin{aligned}
        V_t &\leq \rho^t V_0 +\sum^{t-1}_{k=0}\rho^k \beta_V(\bar{\varepsilon})\leq \rho^t V_0 + \frac{\beta_V(\bar{\varepsilon})}{1-\rho}.
    \end{aligned}
\end{equation}
We then apply \cref{lower and upper bound lemma}. Specifically, we apply the lower bound to the left-hand side and apply the upper bound to the right-hand side. Consequently,
\begin{equation}
    \|\xi_t\|^2 \leq \frac{c_{V,1}}{\gamma \lambda_{\min}(P_W)}\rho^t \|\xi_0\|^2 + \frac{\rho^tc_{V,2}(\bar{\varepsilon})+\frac{\beta_V(\bar{\varepsilon})}{1-\rho}}{\gamma\lambda_{\min}(P_W)}.
\end{equation}
By using $\sqrt{a+b}\leq \sqrt{a}+\sqrt{b}$, we obtain
\begin{equation}
    \|\xi_t\| \leq c \left(\sqrt{\rho}\right)^t\|\xi_0\|+d_t,
\end{equation}
where
\begin{equation}
    c = \sqrt{\frac{c_{V,1}}{\gamma \lambda_{\min}(P_W)}}, d_t = \sqrt{\frac{\rho^t c_{V,2}(\bar{\varepsilon})+\frac{\beta_V(\bar{\varepsilon})}{1-\rho}}{\gamma\lambda_{\min}(P_W)}}.
\end{equation}
Since $d_t$ can be upper-bounded by a constant using $\rho^t\leq1$, this result implies practical exponential stability. Moreover, we can derive
\begin{equation}
    \limsup\limits_{t\to\infty} \|\xi_t\| \leq \sqrt{\frac{\beta_V(\bar{\varepsilon})}{(1-\rho)\gamma\lambda_{\min}(P_W)}}.
\end{equation}
Together with $\|x_t\| \leq \|\Gamma\| \|\xi_t\|$, we conclude that the system state ultimately converges to a ball around the origin whose radius vanishes as $\bar{\varepsilon}\rightarrow 0$. The proof is complete. $\hfill \square$

\footnotesize{ 
\bibliographystyle{IEEEtran}
\bibliography{reference}
}
\newpage
\begin{IEEEbiography}[{\includegraphics[width=1in,height=1.25in,clip,keepaspectratio]{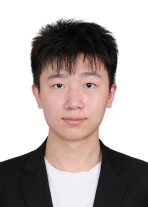}}]{Zhaohua Yang} received the B.Eng. degree from Zhejiang University, Hangzhou, China, and the B.S. degree from the University of Illinois at Urbana-Champaign, Urbana, USA, in 2022. He is currently pursuing the Ph.D. degree in the Department of Electronic and Computer Engineering at the Hong Kong University of Science and Technology, Hong Kong, China. From August 2021 to January 2022, he was a visiting student at the Grainger College of Engineering, the University of Illinois at Urbana-Champaign, Urbana, USA.
		His research interests include optimal control, data-driven control, and model predictive control.
	\end{IEEEbiography}
    \vspace{-2\baselineskip}
\begin{IEEEbiography}[{\includegraphics[width=1in,height=1.25in,clip,keepaspectratio]{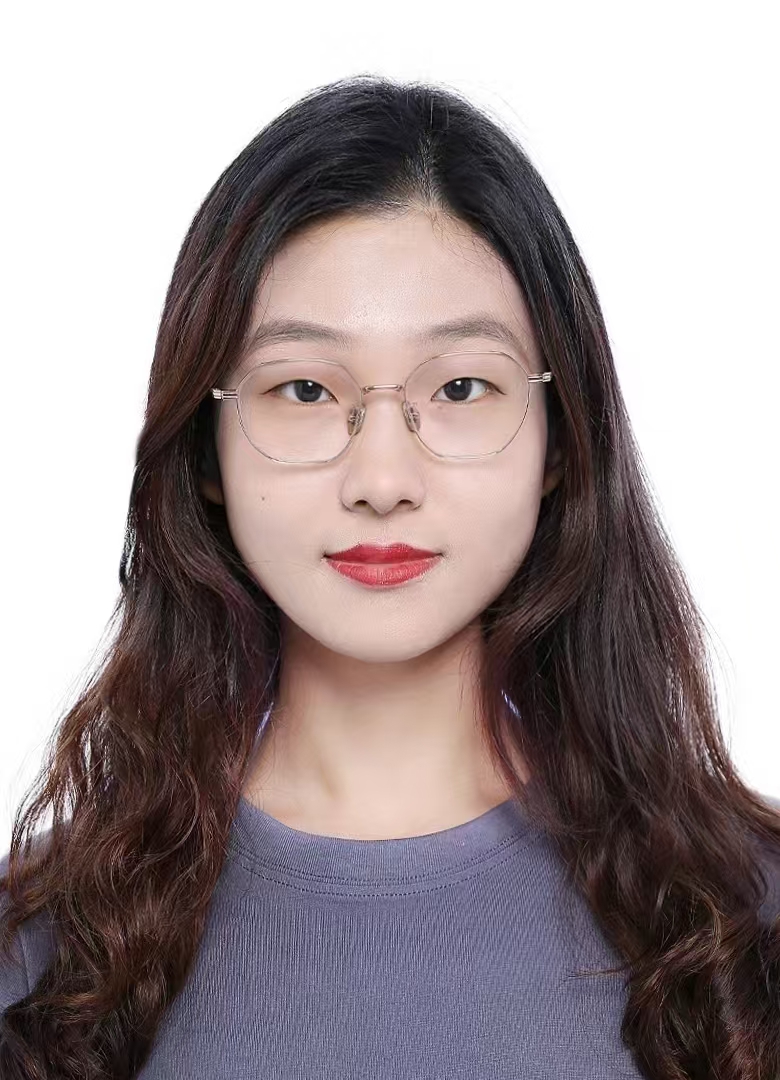}}]{Nan Bai} received the B.S. degree in Mathematics from Northeastern University, Shenyang, China, in 2020, and the Ph.D. degree in Dynamical Systems and Control from the Peking University, Beijing, China, in 2025. She is currently a post-doctoral fellow in the Department of Electronic and Computer Engineering, The Hong Kong University of Science and Technology, Clear Water Bay, Kowloon, Hong Kong, China.
Her research interests include cooperative control of multi-agent systems, distributed optimization, and model predictive control.
	\end{IEEEbiography}
\vspace{-2\baselineskip}
\begin{IEEEbiography}[{\includegraphics[width=1in,height=1.25in,clip,keepaspectratio]{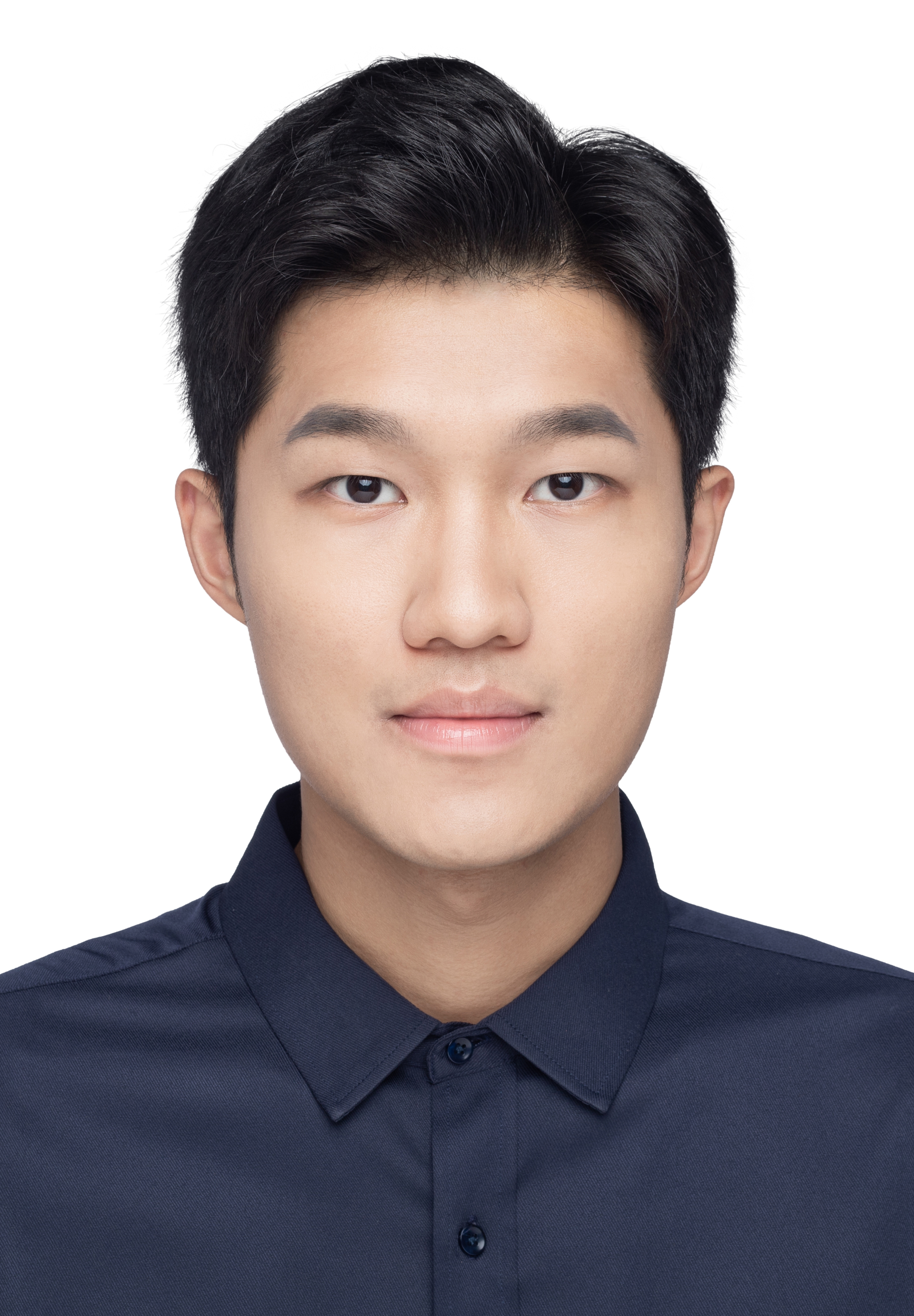}}]{Pengyu Wang} (Member, IEEE) received the B.E. degree in Electronic Engineering from Shandong University, Qingdao, China, in 2022, and the Ph.D. degree in Electronic and Computer Engineering from the Hong Kong University of Science and Technology (HKUST), Hong Kong, in 2026. He was a Visiting Scholar at Tsinghua University, Beijing, China, University of Victoria, Victoria, Canada, and Cornell University, Ithaca, USA. He is currently a Postdoctoral Fellow with the Department of Electronic and Computer Engineering, HKUST. His research interests include multi-robot systems, task scheduling and planning, field robotics, and the safety and security of robotic systems.
	\end{IEEEbiography}
\vspace{-2\baselineskip}
\begin{IEEEbiography}[{\includegraphics[width=1in,height=1.25in,clip,keepaspectratio]{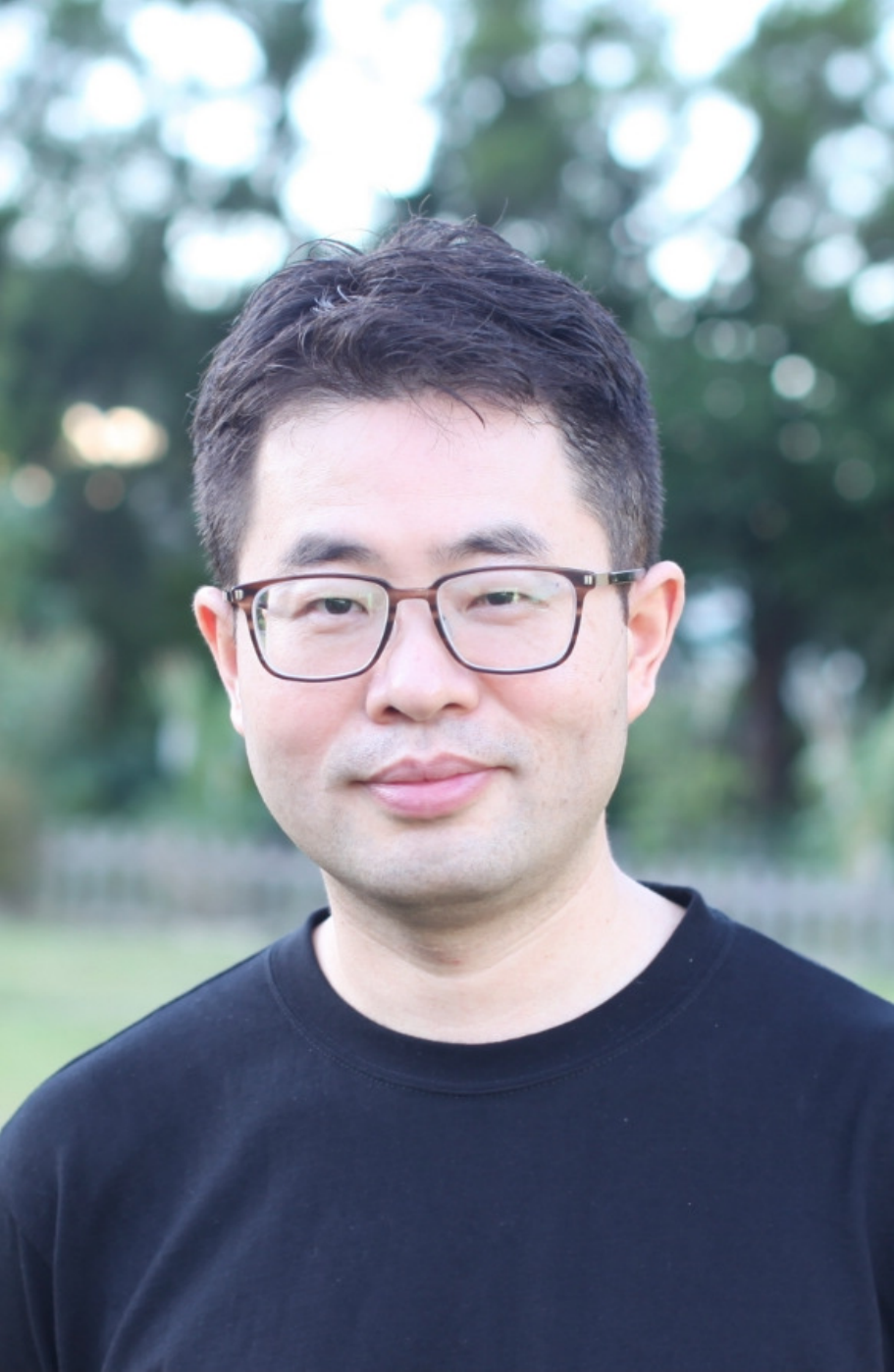}}]{Ling Shi} received his B.E. degree in Electrical and Electronic Engineering from The Hong Kong University of Science and Technology (HKUST) in 2002 and Ph.D. degree in Control and Dynamical Systems from The California Institute of Technology (Caltech) in 2008. He is currently a Professor in the Department of Electronic and Computer Engineering at HKUST with a joint appointment in the Department of Chemical and Biological Engineering (2025-2028), and the Director of The Cheng Kar-Shun Robotics Institute (CKSRI). His research interests include cyber-physical systems security, networked control systems, sensor scheduling, event-based state estimation, and multi-agent robotic systems (UAVs and UGVs). He served as an editorial board member for the European Control Conference 2013-2016. He was a subject editor for International Journal of Robust and Nonlinear Control (2015-2017), an associate editor for IEEE Transactions on Control of Network Systems (2016-2020), an associate editor for IEEE Control Systems Letters (2017-2020), and an associate editor for a special issue on Secure Control of Cyber Physical Systems in IEEE Transactions on Control of Network Systems (2015-2017). He also served as the General Chair of the 23rd International Symposium on Mathematical Theory of Networks and Systems (MTNS 2018). He is currently serving as a member of the Engineering Panel (Joint Research Schemes) of the Hong Kong Research Grants Council (RGC) (2023-2026). He received the 2024 Chen Han-Fu Award given by the Technical Committee on Control Theory, Chinese Association of Automation (TCCT, CAA). He is a member of the Young Scientists Class 2020 of the World Economic Forum (WEF), a member of The Hong Kong Young Academy of Sciences (YASHK), and he is an IEEE Fellow.
	\end{IEEEbiography}
\end{document}